%% file: ghquant.tex
\documentclass[11pt,a4paper]{article}
\usepackage[utf8]{inputenc}
\usepackage[english]{babel}
\usepackage[T1]{fontenc} 
\usepackage[nottoc]{tocbibind}
\usepackage{authblk}
\usepackage{float}            
\usepackage{amsmath}  
\usepackage{mathrsfs}
\usepackage{amsfonts}         
\usepackage{amssymb}          
\usepackage{amsthm}    
\usepackage{enumerate}        
\usepackage[bottom]{footmisc} 
\usepackage[square,numbers]{natbib}	  

\usepackage[colorlinks,pdfpagelabels,pdfstartview = FitH,bookmarksopen = true,bookmarksnumbered =
  true,linkcolor = black,plainpages = false,hypertexnames = false,citecolor = black]{hyperref}

\usepackage[a4paper,textwidth=17cm,textheight=24cm,headheight=14pt]{geometry}

\allowdisplaybreaks

\theoremstyle{plain}
\newtheorem{theorem}{\textsc{Theorem}}[section]
\newtheorem{lemma}[theorem]{\textsc{Lemma}}
\newtheorem{proposition}[theorem]{\textsc{Proposition}}
\newtheorem{corollary}[theorem]{\textsc{Corollary}}
\theoremstyle{definition}
\newtheorem{definition}[theorem]{\textsc{Definition}}
\newtheorem{example}[theorem]{\textsc{Example}}  
\newtheorem{remark}[theorem]{\textsc{Remark}}
\newtheorem{assumption}[theorem]{Assumption}

\newcommand{\R}{\mathbb{R}}
\newcommand{\bC}{\mathbb{C}}

\newcommand{\D}{\mathbb{D}}

\newcommand{\im}{\mathrm{i}}

\usepackage{fancyhdr}
\numberwithin{equation}{section}

\title{Complex Structures for Klein-Gordon Theory\\ on Globally Hyperbolic Spacetimes}
\author[1]{Albert Much\footnote{much@itp.uni-leipzig.de}}
\author[2]{Robert Oeckl\footnote{robert@matmor.unam.mx}}
\affil[1]{Institut für Theoretische Physik, Universität Leipzig, D-04103 Leipzig}
\affil[2]{Centro de Ciencias Matemáticas,
	Universidad Nacional Autónoma de México,
	C.P. 58190, Morelia, Michoacán, Mexico}
\date{29 November 2018\quad 13 June 2019 (v2)\quad 20 January 2021 (v3)\\ 22 July 2021 (v4) \quad 4 January 2022 (v5) \\ UNAM-CCM-2018-3}

\begin{document}
\maketitle
\abstract{\input{abstract}}

\tableofcontents

\newpage

\input{intro}

\input{ghyp}

\input{cstruct}

\input{domains}

\input{natcomplex} 
\input{ackn} 
\appendix

\input{proofs}

\newcommand{\eprint}[1]{\href{https://arxiv.org/abs/#1}{#1}}
\bibliographystyle{stdnodoi}
\bibliography{stdrefsb,allliterature1}

\end{document}

%% file: abstract.tex
We develop a rigorous method to parametrize complex structures for Klein-Gordon theory in globally hyperbolic spacetimes that satisfy a completeness condition. The complex structures are conserved under time-evolution and implement unitary quantizations. They can be interpreted as corresponding to global choices of vacuum. The main ingredient in our construction is a system of operator differential equations. We provide a number of theorems ensuring that all ingredients and steps in the construction are well-defined. We apply the method to exhibit natural quantizations for certain classes of globally hyperbolic spacetimes. In particular, we consider static, expanding and Friedmann-Robertson-Walker spacetimes. Moreover, for a huge class of spacetimes we prove that the  differential equation for the complex structure is given by the Gelfand-Dikki equation.

%% file: intro.tex
\section{Introduction}

The first mathematically rigorous work to formulate quantum field theory (QFT) on   globally hyperbolic and stationary spacetimes, that admit non-compact Cauchy surfaces, used the notion of a one-particle-structure \cite{K78}. Developing this further leads to the notion of \emph{complex structure}.\footnote{A rigorous construction of complex structures for globally hyperbolic and stationary spacetimes, that admit  compact Cauchy surfaces was provided in \cite{MCI}.} A complex structure is a symplectic  transformation, i.e., a map on phase space, that leaves the symplectic structure of the system invariant. Moreover, the square of this transformation is minus the identity.
In the case of a linear quantum field on Minkowski space, one would 
take the eigenspaces of the complex structure with eigenvalues $\pm\im$ to be the positive and negative frequency solutions of the classical field equation, and build up the Fock representation from this splitting as usual.  One can 
proceed analogously in the general case, calling solutions ``positive'' or 
``negative frequency'' according to their decompositions under the 
complex structure. This leads to a familiar ladder-algebra of 
operators, with analogous mathematical notions of ``vacuum'' and ``particle'' states. The terms ``vacuum'', ``particle'', ``annihilation operator'' and ``creation operator'' are used in this mathematical sense in the following.  The questions of what a physical detector would respond to, and what a physically 
acceptable vacuum would be, are outside the scope of this article, 
which is focussed on a consistent mathematical treatment.

For   globally hyperbolic spacetimes that are not static nor stationary (nor conformally related to such spacetimes) the complex structure   appears to be  time-dependent \cite[Page7]{AsMa:qfieldscurved}. In particular, if the complex structure is viewed as acting on a fixed space of initial data this frequently leads to a failure of unitary implementability of the dynamics in the quantum theory. This problem has even been framed in terms of a no-go theorem \cite{ToVa:fevolqft}.
On the other hand, in more recent work, a complex structure is associated instead with each hypersurface, in a coherent fashion \cite{Oe:holomorphic,AA15}. This idea is natural considering the analogy to the symplectic structure. The latter, even though constant in the usual construction of a fixed initial data space, originates from an integral over the second variation of the Lagrangian on a spacelike hypersurface \cite{Woo:geomquant}. In the present work we consider the strongest condition for the coherence of the complex structures (see Equation~(\ref{jhs})), leading to the notion of \emph{conserved complex structure}, corresponding to particle number conservation. (Again, particle number conservation is relative to the technical notion of vacuum in terms of creation and annihilation operators and does not necessarily imply that a particular model detector would not see any particle creation.) Unitary implementability of the dynamics in the quantum theory is then automatic. (In \cite{AA15} this is termed \emph{generalized unitarity}. For more details, see Section~\ref{sec:cstruct}.) 
In the present work we take advantage of this more flexible approach to address the quantization problem for Klein-Gordon theory in globally hyperbolic spacetime. To this end we associate a phase space to each leaf of a temporal foliation of spacetime. There, we consider the complex structure as a two by two matrix of operators with respect to a canonical decomposition \cite{AsMa:qfieldscurved,CCQ:schroefock,Oe:schroedhol,AA15}. We then formulate the coherence condition for the complex structures in terms of operator differential equations that we term the \emph{conservation equations}. We provide quantizations of Klein-Gordon theory on various classes of globally hyperbolic spacetimes by exhibiting natural solutions of the conservation equations. In particular, we generalize and make rigorous previous results for static spacetimes and provide new quantizations for expanding and Friedman-Robertson-Walker (FRW) spacetimes. 
For the (well known) static   and the expanding spacetimes the resulting complex structures do  not depend on the time (see Theorem~\ref{t0} and Theorem~\ref{t1}). However, for FRW spacetimes there is an explicit time dependence in the complex structure (see Theorem~\ref{t4}) that however agrees by construction, as mentioned before,   with the unitarity of time-evolution. 

We prove that our constructed complex structures are rigorously defined. The proof is two-fold. First, we have to prove that the given complex structures are anti-self-adjoint operators. This is done by    using the recent proof of essential self-adjointness of the spatial part of the Klein-Gordon operator \cite{MuOe:selfadjoint}. Secondly, we have to prove that the  constructed complex structures supply automorphisms of the solution spaces of the Klein-Gordon equation. To assure this automorphism property we construct the aforementioned spaces as Sobolev spaces by using the spatial part of the Klein-Gordon operator. 

We anticipate that the methods developed in the present work could be used to find and verify complex structures for other interesting classes of globally hyperbolic spacetimes, such as for example Bianchi, or Kasner-like, universes.
We also recall that the traditional approach (and its failure of a unitary description \cite{K79}) relies heavily on the Hamiltonian. In contrast, in our present framework we eliminate  the reference to a particular Hamiltonian for constructing the complex structure. Rather, we build the time evolution by means of an embedding that relies on the global hyperbolicity. Therefore, we expect our approach to be more generally applicable. However, we emphasize that the methods presented here only allow for particle-number conserving quantizations. We hope to address the more general case involving particle creation in future work.

In Section~\ref{sec:ghyp} we give a short introduction to Klein-Gordon theory on globally hyperbolic spacetimes. The unitarity condition and the conservation equations that determine the complex structure are given in Section~\ref{sec:cstruct}. We dedicate Section~\ref{sec:domains} to the study of the domains and the images of the complex structure in a rigorous fashion. The reader interested only in the physical results may skip this section. Using the results of the previous sections, we give well-defined conserved complex structure for a variety of spacetimes  in Section~\ref{sec:natcomplex}. Finally, we comment on the emergence of the Gelfand-Dikki equation as underlying a large class of examples. In order to ease readability we collect all lengthy proofs in Appendix~\ref{sec:proofs}.

%% file: ghyp.tex
\section{Klein-Gordon Theory on Globally Hyperbolic Spacetimes}\label{sec:ghyp}

A \emph{globally hyperbolic spacetime} is a manifold $(M,\mathbf{g})$ that is homeomorphic \cite{G70} and even diffeomorphic \cite{B05} to the split manifold $\mathbb{R}\times\Sigma$ with Cauchy-surfaces $\{t\}\times\Sigma$, where $t\in\mathbb{R}$. 	The authors in \cite{B03} solved a long-standing conjecture by proving that any globally hyperbolic spacetime admits a  \emph{smooth} foliation into Cauchy surfaces \cite[Theorem~1.1]{B03}. Moreover, the induced metric of  a globally hyperbolic spacetime admits a specific form \cite[Theorem~1.1]{B05}.

\begin{theorem}\label{T1}	Let $(M, 	\mathbf{g})$ be a globally hyperbolic spacetime. Then, it is isometric to
	the smooth product manifold $\mathbb{R}\times \Sigma$ with a metric $g$, i.e.,
	\begin{equation}
	g_{\mu\nu}dx^{\mu}dx^{\nu}=-N^2 dt^2+h_{ij}dx^idx^j,
	\label{mgh}\end{equation}
	where $\Sigma$ is a smooth 3-manifold and $(x^0,x^1,x^2,x^3)=(t,x^1,x^2,x^3)$ are Cartesian coordinates on $\R\times\Sigma$. We write the (often implicit) isometry as $T:\R\times \Sigma\to M$ with $T_t:\Sigma\to M$ the parametrized embedding given by $T_t(x) := T(t,x)$. $N: M \mapsto (0,\infty)$ is a smooth function, and $\mathbf{h}$ a $2$-covariant symmetric tensor field on $M$, satisfying the following condition:  Each hypersurface  $\Sigma_t := T_t(\Sigma)$ at constant $t$ is a Cauchy surface, and the restriction $\mathbf{h}(t)$ of $\mathbf{h}$ to such a  $\Sigma_t$ is a Riemannian metric (i.e.\  $\Sigma_t$ is spacelike).
\end{theorem}  

Next, we give for a spin zero scalar field in a globally hyperbolic spacetime the equation of motion and the corresponding symplectic structure. These  are the classical ingredients for a subsequent quantization.  In the following sections those structures form the basis for finding a well-defined natural complex structure. This in turn allows to build the Hilbert space of states of the quantum theory. The  Klein-Gordon equation  is,
\begin{equation}\label{kge}(\square_{g}-m^2-V(x))\phi=0,\end{equation} 
where $\square_{g}$ is the Laplace-Beltrami operator for the metric $\mathbf{g}$, i.e. $\square_{g}=  ({\sqrt{|g|}})^{-1}\partial_a(\sqrt{|g|}g^{ab}\partial_b)$ with $|g|$ denoting the absolute value of the determinant of the metric $\mathbf{g}$.

\begin{proposition}\label{p1}
	Let the operator  $w^2$ be given by   \begin{align}\label{op}
	w^2&=-\frac{N}{\sqrt{|h|}}\partial_i(
	\sqrt{|h|} N h^{ij}\partial_j
	)+N^2m^{2}+N^2\,V\\\nonumber&=
	-N^2(\Delta_{h}-m^{2}-V)-Nh^{ij}\partial_iN\partial_j .
	\end{align}		
$\Delta_h$ is the Laplace-Beltrami operator for the associated metric $\mathbf{h}$ on $\Sigma_t$. Moreover, let $f$ be 	defined as $f:=-N^{-1} ( \partial_t N ) +  (\sqrt{|h|})^{-1}  \,(\partial_t\sqrt{|h|})$. Then, for the explicit form  (\ref{mgh}) of a globally hyperbolic metric  we have the following Klein-Gordon equation,
	\begin{align}  
	\left( 	 \partial^{2}_t    +   f\partial_t+     w^2 \right)\phi = 0 . \label{eq:kgt}
	\end{align}  
\end{proposition}
\begin{proof}
	See Appendix \ref{pp1}.
\end{proof}

Conditions and a domain for which the operator $w^2$ is essentially self-adjoint, positive and invertible are given in \cite{MuOe:selfadjoint}. See \cite{K78} for a proof for stationary spacetimes and \cite{BF} for a proof for cases of FRW spacetimes.

Next, we fix a time $t_0$ and use $T_{t_0}$ to identify the Cauchy surface $\Sigma_{t_0}$ with $\Sigma$. Denote the space of smooth initial data of compact support for $\Sigma$  by $\mathscr{S}(\Sigma)$.\footnote{In Section \ref{sec:domains} we specify (construct) the explicit form of the space $\mathscr{S}(\Sigma)$.}  
A vector $(\varphi,\pi)\in\mathscr{S}(\Sigma)$ of the initial data space at time $t_0$ defines a unique solution (due to Leray's Theorem \cite{LJ}) $\phi:M\rightarrow \mathbb{R}$. The relation is as follows,
\begin{align}
\begin{pmatrix}\varphi\\ \pi\end{pmatrix}
= \begin{pmatrix}\phi\\
{\sqrt{|h|}}n^a\nabla_a\phi\end{pmatrix}\bigg\rvert_{t=t_0}
= \begin{pmatrix}\phi\\
{\sqrt{|h|}}N^{-1}\partial_t\phi\end{pmatrix}
\bigg\rvert_{t=t_0},
\label{eq:lpsident}
\end{align}
where $n^{a}$ is an arbitrary timelike hypersurface-orthogonal vector field and $\nabla_a$ the corresponding covariant derivative. Furthermore, the \emph{symplectic form}  corresponding to the defined initial data space $\mathscr{S}({\Sigma})$
is the map $\Omega_{\Sigma}:\mathscr{S}({\Sigma}) \times \mathscr{S}({\Sigma}) \mapsto \mathbb{R}$ given by,
\begin{align}
\Omega_{\Sigma} ( (\varphi_1,\pi_1), (\varphi_2,\pi_2)) &=\int_{\Sigma}\left(
\pi_1 \varphi_2 - \pi_2\varphi_1 
\right)d^3x   .
\label{eq:symplectic}
\end{align}
The pair $(\mathscr{S}({\Sigma}),\Omega_{\Sigma})$ is a \emph{symplectic vector space}. Note that the right-hand side of (\ref{eq:symplectic}) has neither an explicit dependence on $t_0$ nor on the metric $\mathbf{h}(t_0)$, due to our particular choice of data in (\ref{eq:lpsident}).

For the forthcoming sections we adopt the definition of an isometry on Riemannian Manifolds, see for example \cite{THC}. In case $\mathbf{g}_1$ and $\mathbf{g}_2$ are the chosen Lorentzian
metrics on $M_1$ and $M_2$ and $\chi^*\mathbf{g}_1 = \mathbf{g}_2$, we call $\chi$ an \emph{isometry}; if $\chi^*\mathbf{g}_1 =\Lambda^2\mathbf{g}_2$ with a
strictly positive smooth function $\Lambda$, $\chi$ is called a \emph{conformal isometry} and $\Lambda^2\mathbf{g}_2$ a conformal transformation of $\mathbf{g}_1$.

%% file: cstruct.tex
\section{Complex Structures and Unitarity}\label{sec:cstruct}

We may construct a Hilbert space of states by adding to the classical data a \emph{complex structure}. That is, we equip the initial data space $\mathscr{S}({\Sigma})$ with an operator $J:\mathscr{S}({\Sigma})\to\mathscr{S}({\Sigma})$ that is a complex structure, i.e., satisfies $J^2=-1$. Moreover, $J$ is required to be compatible with the symplectic structure, i.e., it must be a \emph{symplectic transformation} $\Omega_{\Sigma}(J\Phi_1,J\Phi_2)=\Omega_{\Sigma}(\Phi_1,\Phi_2)$. Here, as in the following, we use the compact notation $\Phi=(\varphi,\phi)$ for elements of $\mathscr{S}(\Sigma)$. Finally, we require $J$ to \emph{tame} the symplectic structure, i.e., the hermitian sesquilinear form
  \begin{equation}
    \{\Phi_1,\Phi_2\}=2\Omega_{\Sigma}(\Phi_1,J\Phi_2)+2\im\Omega_{\Sigma}(\Phi_1,\Phi_2), 
\label{cip}
  \end{equation}
has to be positive-definite and thus define a complex inner product. After completion this yields a Hilbert space $\mathscr{H}_J$. The Hilbert space of states of the quantum theory is constructed as the \emph{Fock space} over this inner product space.

A more traditional, but equivalent way to construct a quantization \cite{BiDa:qftcurved} is to start with the following non-degenerate sesquilinear and hermitian form on the space of complexified solutions,
\begin{equation}
  (\Phi_1,\Phi_2)=4\im\Omega_{\Sigma}(\overline{\Phi_1},\Phi_2) .
  \label{eq:stdip}
\end{equation}
Then choose a complete set $\{u_k\}_{k\in I}$ of ``positive energy'' solutions satisfying,
\begin{equation}
  (u_k,u_l)=\delta_{k,l},\quad (\overline{u}_k,\overline{u}_l)=-\delta_{k,l},\quad
  (u_k,\overline{u}_l)=0 .
\end{equation}
The relation between this approach and that using a complex structure is as follows. Given a complete set of solutions $\{u_k\}_{k\in I}$ with the specified properties, we define $J$ as the operator with eigenvalue $\im$ on the space spanned by the solutions $\{u_k\}_{k\in I}$ and with eigenvalue $-\im$ on the space spanned by the solutions $\{\overline{u}_k\}_{k\in I}$. Conversely, given $J$ we choose an orthonormal basis $\{u_k\}_{k\in I}$ of the eigenspace with eigenvalue $\im$ of $J$ with respect to the inner product (\ref{eq:stdip}). The advantage of using the complex structure over the explicit solutions is that we do not need to make an explicit choice of basis (on which the quantization does not depend).

In order to have explicit expressions for the complex structure we use operators that are naturally induced by the Klein-Gordon equation (as for example in \cite{AsMa:qfieldscurved,K78,K79}).  On the symplectic vector space $(\mathscr{S}(\Sigma),\,\Omega_{\Sigma})$  with coordinates $(\phi,\pi)_{\Sigma}$, the complex structure $J$ may be parametrized as \cite{AsMa:qfieldscurved, CCQ:schroefock},
\begin{equation}
    -J=\left(
    \begin{matrix}
        A&B\\D&C
    \end{matrix}
    \right)   \label{eq:jabcd}
\end{equation} 
where since $J^2=-1$ the linear operators $A,\,B,\, C,\, D$ satisfy the following relations, 
   \begin{equation*}
   A^2+BD=-1,\qquad  C^2+DB=-1,  \qquad AB+BC=0,\qquad   DA+CD=0, 
   \end{equation*}
which leads to
\begin{equation}		
	    D= -B^{-1}(1+A^2) 	\qquad C= -B^{-1} A B .
            \label{eq:cd}
\end{equation}
 Due to the symmetry and positive definiteness of  the inner product $\{\cdot,\cdot\}$  the linear operators $A,B,C,D$ have to satisfy the following conditions,  
  \begin{align}\label{sop}
 \int_{\Sigma}
 \Psi B \Psi'=\int_{\Sigma}
 \Psi' B \Psi,\qquad  \int_{\Sigma}
\chi D \chi'=\int_{\Sigma}
 \chi' D \chi, \qquad 
  \int_{\Sigma}
 \Psi A \chi=-\int_{\Sigma}
\chi C \Psi,
\end{align}
\begin{align*}
\int_{\Sigma}
\Psi B \Psi >0,\qquad\qquad\qquad  \int_{\Sigma}
\chi D \chi <0,
\end{align*}
where $\chi,\chi'$ are scalars that are elements in the space $C_0^{\infty}(\Sigma)$ and $\Psi,\Psi'\in C_0^{\infty}(\Sigma)$ scalar densities of weight one.

Different complex structures may give rise to inequivalent quantizations. In the sense of representation theory, two complex structures $J_1, J_2$ on the same phase space give rise to equivalent quantizations if and only if their difference $J_1-J_2$ is a Hilbert-Schmidt operator on the Hilbert space determined by one of them via (\ref{cip}) \cite{Sha:linsymboson}. On the other hand, the physical role of the complex structure is usually tied to a notion of time-evolution and energy. Very briefly, if we look at  Equation~(\ref{eq:kgt}) and consider the simple case that $f$ vanishes we want to have $w^2$ positive definite and $\partial_t^2$ consequently negative definite. (The eigenvalues of $\partial_t^2$ are then minus the square of the energy.) In particular, the spectrum of $\partial_t$ is imaginary and the standard complex structure is then chosen to distinguish the positive imaginary (``positive energy'') part from the negative imaginary (``negative energy'') part. In general curved spacetimes the situation is more complicated and the choice of complex structure becomes less straightforward.

We limit ourselves in this context to remark that the construction of a complex structure along such lines would in general depend on time for globally hyperbolic spacetimes that are  non-stationary \cite{AsMa:qfieldscurved,K79}. Proofs for a variety of time-dependent spacetimes exist that natural fixed choices of complex structures are not compatible with unitary time evolution \cite{H1} in the following sense: For candidate time evolution maps $U: \mathscr{H}_{J}\to\mathscr{H}_{J}$, the operator $J-UJU^{-1}$ is not Hilbert-Schmidt in $\mathscr{H}_{J}$.

To resolve this tension we use techniques and ideas developed in \cite{Oe:holomorphic,AA15}. Let us recap some of the most important results that we  use. In the mentioned works one defines a phase space of initial data for each leaf of the foliation. Since the leaves are labeled  by a   time coordinate $t\in\R$ we denote the corresponding phase spaces by $\Gamma_t$. There is also the phase space $\Gamma_V$ of global solutions of the equations of motion. Since we have a well posed initial value problem we have isomorphisms $\mathcal{I}_t: \Gamma_t\rightarrow \Gamma_V$. What is more, each phase space $\Gamma_t$ is naturally identified with $\mathscr{S}(\Sigma)$ due to the embeddings given in Theorem~\ref{T1} of the manifold $\Sigma$ into $M$. Recall that for a time $t$ we denote this embedding by $T_t:\Sigma  \rightarrow M$. With the identification between $\Gamma_t$ and $\mathscr{S}(\Sigma)$ implicit we read off the isomorphism $\mathcal{I}_{t_0}^{-1}:\Gamma_V\to\Gamma_{t_0}$ from Expression~(\ref{eq:lpsident}).

As objects naturally associated with the phase space $\Gamma_t$ for each leaf of the foliation, we denote the symplectic structure by $\Omega_t$ and the complex structure by $J_t$. Indeed, the symplectic structure $\Omega_t$ arises as a second variation of the action on the hypersurface $T_t(\Sigma)$ \cite{Woo:geomquant}. However, due to a conservation law that it satisfies, combined with our special choice of coordinates on momentum space (compare Equation~(\ref{eq:lpsident})) this symplectic form viewed on $\mathscr{S}(\Sigma)$ takes the same form of  Equation~(\ref{eq:symplectic}) for any time $t$. Thus, we use the notation $\Omega_{\Sigma}$ interchangeably with $\Omega_t$. There is no reason to expect the same to happen for the complex structure. That is, there is no reason to expect the complex structures $J_t$ to lead to the same complex structure $J$ on $\mathscr{S}(\Sigma)$ for all times $t$ as we have simply assumed above. Indeed, it is easy to see that such a requirement of ``constancy'' of the symplectic structure is dependent on the choice of coordinate system and Cauchy hypersurface embeddings and thus generically unphysical.

The most general condition for a family of complex structures $\{J_t\}_{t\in\R}$ to admit unitary maps that could be candidates for the quantum time-evolution is (as implied by our previous remarks) that,
\begin{equation} \label{jhs}
J_{t_1}-E^{-1}_{t_2,t_1} \, J_{t_2}\, E_{t_2,t_1}\qquad  \text{is Hilbert-Schmidt in $\mathscr{H}_{J_{t_1}}$}.
\end{equation}
Here the (classical) time evolution map $E_{t_2,t_1}: \Gamma_{t_1}\rightarrow\Gamma_{t_2}$ defined by
\begin{align}\label{tev}E_{t_2,t_1}=\mathcal{I}_{t_2}^{-1}\mathcal{I}_{t_1},
\end{align}
is a map on the canonical phase space that evolves states from time $t_1$ to time $t_2$.  In particular, it takes Cauchy data defined on the Cauchy surface $\Sigma_1$ (which is the embedding of the Cauchy surface $\Sigma$ that results from the foliation at time $t_1$, i.e. $T_{t_1}(\Sigma)=\Sigma_{t_1}$) evolves it to a global solution and induces the Cauchy data on the surface  $\Sigma_{t_2}$.

If Expression~(\ref{jhs}) even vanishes for any pair of times $t_1$ and $t_2$, then the complex structures $J_t$ arise from a complex structure $J_V$ on the global phase space $\Gamma_V$. That is, we have for all times $t$,
\begin{equation} \label{con}
J_t=\mathcal{I}_t^{-1}\, J_V\, \mathcal{I}_t.
\end{equation}
Note that this condition is not the same as the ``constancy'' condition mentioned previously, but rather means that the complex structure is \emph{conserved} under time-evolution. In  all of the remainder of this paper we restrict ourselves to this case. In the standard quantization, this guarantees particle number conservation  from the outset as well as unitarity.

Since the complex structure viewed as an operator on $\mathscr{S}(\Sigma)$ is time dependent, so are the operators $A, B, C, D$ that encode it in Equation~(\ref{eq:jabcd}). We thus write $A(t)$ etc.\ when we need to make this explicit.
Suppose we are given a family of complex structures $\{J_t\}_{t\in\R}$ specified in terms of a corresponding family of operators $A(t)$ and $B(t)$ according to Relation~(\ref{eq:jabcd}). ($C(t)$ and $D(t)$ are redundant because of Relations~(\ref{eq:cd}).) Given a global solution $\phi$ we obtain initial data $(\varphi(t),\pi(t))$ for each time $t\in\R$. Applying at each time $t$ the corresponding complex structure $J_t$ (in terms of the operators $A(t)$ and $B(t)$) leads to new initial data $(\varphi'(t),\pi'(t))$. The family $\{J_t\}_{t\in\R}$ arises from a conserved complex structure, i.e., there exists a corresponding complex structure $J_V$ on the global solution space if and only if the new initial data family $\{(\varphi'(t),\pi'(t))\}_{t\in\R}$ assembles to a new global solution $\phi'$ for any choice of $\phi$. This condition can be formulated as a condition on the family of operators $\{(A(t),B(t))\}_{t\in\R}$. This is how we will characterize conserved complex structures throughout this work: As families $\{(A(t),B(t))\}_{t\in\R}$ that define an admissible complex structure at each time $t$ and in addition satisfy this conservation condition.

We may break the condition down into two parts. After assembling the data $\{\varphi'(t)\}_{t\in \R}$ into a global configuration $\phi'$ we check:
\begin{itemize} 
	\item The Klein-Gordon Equation:	$\left( 	 \partial^{2}_t         -N^{-1} ( \partial_t N ) \partial_t +  (\sqrt{|h|})^{-1}  \,(\partial_t\sqrt{|h|})  \partial_t +     w^2 \right)  \phi'=0$\label{fr1}. \medskip
	\item  The derivative relation: $\pi'(t)=\sqrt{h}N^{-1}\partial_t\phi'|_t$.
\end{itemize}
These conditions are equivalent to the (1-particle) Schrödinger  equation, 
\begin{align}
	\dot{J} = [H,J], 
	\label{eq:1pevolJ}
\end{align}
where the Hamiltonian is given by, 
\begin{align*}H=\left(
	\begin{matrix}
		{0}&{\sqrt{h}N^{-1}}\\{-(\sqrt{h})^{-1}Nw^2}&{-f}
	\end{matrix}\right).
\end{align*}
The equivalence can be seen by taking the time derivative of the term $J\Phi$, i.e.\
\begin{align*}
	\frac{d}{dt}(J\Phi)= H(J\Phi) \end{align*}
and using the Klein-Gordon equation for the initial data $\Phi$, i.e.\ $\dot{\Phi} = H\Phi$. If we use the non-weighted initial data\footnote{The matrix $T$ in Lemma \ref{lcsme} gives the similarity transformation from weighted to non-weighted initial data.}, i.e. $\Psi=(\varphi,\dot{\varphi})$, the Hamiltonian $H$ reads,
\begin{align*}H=\left(
	\begin{matrix}
		{0}&{1}\\{-w^2}&{-f}
	\end{matrix}\right).
\end{align*}

\begin{proposition}\label{p2}
	Demanding that the field $\phi'$ is a solution to the Klein-Gordon equation and that the derivative relation $\pi'(t)=\sqrt{h}N^{-1}\partial_t\phi'|_t$ holds equivalent to the equation of motion $\frac{d}{dt}(J\Phi)= H (J\Phi)$ induces	 the two \textbf{Conservation Equations},\newline
	\begin{enumerate}[(I)] 
		\item $\partial_t Z = - Y w^2 +Y^{-1}(1+Z^2),$\label{me3}\newline
		\item $\partial_tY  = Y^{-1} Z Y+ Yf+ Z,$\label{me4}\newline
	\end{enumerate}
	where we defined $Y:=B\,N^{-1}(\sqrt{h})$. 
	\begin{proof}
		See Appendix~\ref{pp2}.
	\end{proof}
\end{proposition}

%% file: domains.tex
\section{Domains, Images and all that}
\label{sec:domains}

Besides the mentioned requirements on the complex structure, i.e., $J^2=-1$ and the symmetry properties (that give us the tameness property, see Equations~(\ref{sop})), we have to demand from the operator-valued matrix $J:\mathscr{S}( {\Sigma})\rightarrow \mathscr{S}({\Sigma})$ to be an anti-self-adjoint operator $J^{*}=-J$ on the Hilbert space $\mathscr{H}_{J}$. This insures that the constructed complex inner product (see Equation~(\ref{cip})) defined by the use of a complex structure is well-defined. 
The symmetry properties expressed in Equations~(\ref{sop}) are not enough for the operator $J$ to be  anti-self-adjoint. Thus, we need a more rigorous investigation of the domains, for the respective operators that build the object $J$.  Moreover, the complex structure has   to have the automorphism property  w.r.t.\ the solution space $\mathscr{S}({\Sigma})$. From the general static solutions in   Section~\ref{sec:natcomplex} and from the static case in \cite{K78} we know that the complex structure depends on the square root of  the spatial part of the Klein-Gordon operator (see Equation~(\ref{op})). Hence, to prove that the complex structure has an automorphism property on solution spaces, we need to expand the space of smooth functions with compact support, since   the square root of a (weighted) Laplacian (plus a potential) is not an automorphism on those spaces, see \cite{shand}. The reader that is more interested in the physical results may skip this entire section.

\subsection{Adjointness of the Complex Structure}

In this subsection we study the properties of the explicit domains of the operators that build the complex structure. The investigation   allows us  to make statements about minimal requirements on the  domains to guarantee the anti-self-adjointness of the complex structure.
We  introduce an auxiliary inner product, denoted by $\langle\cdot,\cdot\rangle$. This  inner product is written in terms of the $2\times2$-matrix
$$\varepsilon=\begin{pmatrix} 
0&  1  \\ -1   &  0
\end{pmatrix},$$
and the symplectic form, i.e.,
\begin{align}\label{ssm}\Omega ( \Phi, \Psi)=
\langle\Phi,\varepsilon\Psi\rangle_{L^2(\Sigma )\oplus L^2(\Sigma )},\qquad \qquad    \Phi, \Psi\in\mathscr{S}({\Sigma }) .
\end{align}
This allows to translate the compatibility condition with the symplectic structure into an adjointness condition (as in \cite[Equation~(2.2)]{K78}), 
\begin{align*}
\Omega ( \Phi,J \Psi)&=
\langle\Phi,\varepsilon J\Psi\rangle_{L^2(\Sigma )\oplus L^2(\Sigma )}\\&
= \langle J^{\dagger}\varepsilon^{T}\Phi, \Psi\rangle_{L^2(\Sigma )\oplus L^2(\Sigma )}
\\&
=- \langle J^{\dagger}\varepsilon^{T}\Phi,\varepsilon^{2} \Psi\rangle_{L^2(\Sigma )\oplus L^2(\Sigma )}
\\&
= -\langle \varepsilon^{T}J^{\dagger}\varepsilon^{T}\Phi,\varepsilon  \Psi\rangle_{L^2(\Sigma )\oplus L^2(\Sigma )}\\&
=   \langle J^{*}\Phi,\varepsilon  \Psi\rangle_{L^2(\Sigma )\oplus L^2(\Sigma )}
\\&=  \Omega ( J^{*}\Phi, \Psi) ,
\end{align*}
where the adjoint of the operator-valued matrix $J$ is 
\begin{align}\label{adcs}
J^{*}= \varepsilon \,J^{\dagger}\varepsilon^{T}.
\end{align}
By using the auxiliary scalar product (see Equation~(\ref{ssm})) we are able to   define  anti-self-adjointness of the complex structure $J$. Since the complex structure consists of the operators $A$ and $B$ we elaborate in what follows the restrictions on the respective operators for anti-self-adjointness of the operator-valued matrix $J$ to hold. Yet,   from the results in the next section we know that the complex structure is a function  of the  operator $w^2$ (see Equation~(\ref{op})). Since  essential self-adjointness of the spatial part of the Klein-Gordon operator is proven for the measure $d\mu=N^{-1}(\sqrt{h})d^3x$ (see for example \cite{MuOe:selfadjoint,K78,BF}), we give in the following the transformation that changes the measure. First, we  introduce the notion of weighted Hilbert spaces \cite[Chapter~3.6, Definition~3.17]{AG1}.

\begin{definition}
	A triple $(\Sigma,	\textbf{h},\mu)$ is called a \emph{weighted manifold}, if $(\Sigma,	\textbf{h})$ is a Riemannian manifold and $\rho$ is a smooth positive density function such that the corresponding measure (volume element on $\Sigma$) $\mu$ is given by $d\mu=\rho\, d\Sigma$. 
	Equipped with the measure $\mu$   the \emph{weighted Hilbert space}, denoted as $L^2(\Sigma, \mu)$, is given  as the space of all square-integrable functions   w.r.t.\ the measure $\mu$.  
\end{definition}

In the following we transform the complex  structure w.r.t.\ the measure $\mu$. Afterwards, we  give necessary conditions for the complex structure to be anti-self-adjoint.
\begin{lemma}\label{lcsme}
	Let $J:\mathscr{S}({\Sigma})\rightarrow \mathscr{S}({\Sigma})$ denote the complex structure w.r.t.\ the measure on $\Sigma$.
	Denote  the  complex structure w.r.t.\ the weighted measure $d\mu=N^{-1}(\sqrt{h})d^3x$  on $\Sigma$, which is obtained by transforming all functions to scalars (of density zero), by $J_{Y}:\mathscr{S}( {\Sigma}, {\mu})\rightarrow \mathscr{S}( {\Sigma}, {\mu})$. Then, the operator-valued matrix  $J_{Y}$ is given by the following similarity transformation to $J$,
	\begin{align}\label{tme}
	J_{Y}= T  J T ^{-1},
	\end{align}
	where the transformation matrix $T:\mathscr{S}({\Sigma})\rightarrow \mathscr{S}({\Sigma},{\mu})$ is \begin{align}\label{tmat}
	T =	\begin{pmatrix} 
	1&  0 \\ 0    &   N\sqrt{|{h}|}^{-1}
	\end{pmatrix}.
	\end{align} 	The explicit form of the complex structure is thus given by 	\begin{align}  \label{csm}
	J_{Y }=\begin{pmatrix} 
	A&  Y \\ -Y^{-1} (1+A^2)  & -Y^{-1}AY
	\end{pmatrix}  ,
	\end{align} 
	where we defined $Y:=B\,N^{-1}(\sqrt{h})$ and $A:=-Z$.
\end{lemma} 
\begin{proof}
 One proves this by writing out $\Omega (\Phi_1,J \Phi_2)$ explicitly, changing the measure and substituting for the conjugate momenta $\pi_a=(\sqrt{h})N^{-1}\partial_t\phi_a$, with the vector $\Phi_a=(\phi_a, \pi_a)$ for $a=1,2$. The end result ($Y$) acts on the time derivative of the field, i.e.\ $\partial_t\phi_a$.
\end{proof}
For the complex structure to be anti-self adjoint the next assumption is essential. 
\begin{assumption}\label{ass1}
	Let the operators $A$ and $Y$ given in the previous section have the following   properties w.r.t.\ their adjoint,
	$$Y=Y^*,\qquad D(Y)=D(Y^*),\qquad\qquad \qquad A^*=Y^{-1}AY,\qquad D(A^*)=D(Y^{-1}AY).$$
	In particular this means that the symmetric operator $Y$ is self-adjoint and that the operator $A$ is self-adjoint if it commutes with $Y$.  Moreover, we assume  that the operators $Y$ and $A$ and their respective adjoints all have a common dense and stable domain of  (self-)adjointness, denoted by $\mathcal{D}(\Sigma, \mu)$. It satisfies  $\mathcal{D}(\Sigma, \mu)=D(Y)=D(Y^*)=D(A^*)=D(Y^{-1}AY)$ and moreover has the inclusion  properties  $C^{\infty}_0(\Sigma) \subset \mathcal{D}(\Sigma, \mu)\subset L^2(\Sigma, \mu)$. Furthermore, we make the assumption that the operator $Y$ is invertible  on the domain   $\mathcal{D}(\Sigma, \mu)$. The domain is also stable under the action of the inverse of the operator $Y$ . 
\end{assumption}
In the following section where we discuss solutions for the complex structure on different globally hyperbolic spacetimes we   have to prove for every solution of $A$ and $Y$ that the assumption holds.

\begin{proposition}\label{cssa}
	Let the complex structure $J_{Y}:\mathcal{S}(\Sigma, \mu)\rightarrow\mathcal{S}(\Sigma, \mu)$ be given by the operator-valued matrix
	\begin{align}  
	J_{Y }=\begin{pmatrix} 
	A&  Y \\ -Y^{-1} (1+A^2)  & -Y^{-1}AY
	\end{pmatrix},
	\end{align} 
	and let the operators $A$ and $Y$ satisfy Assumption~\ref{ass1}. Then, the complex structure $J_{Y}$ is anti-self-adjoint on the domain $\mathcal{S}(\Sigma, \mu):=\mathcal{D}(\Sigma, \mu)\oplus \mathcal{D}(\Sigma, \mu)$.
\end{proposition} \begin{proof}
	See Appendix~\ref{p41}.
\end{proof}

\subsection{Complex Structure as an Automorphism}
\label{csauto}

Besides the proof of anti-self-adjointness of the complex structure we have to guarantee that the map (induced by the complex structure) is a topological linear automorphism of the solution space $\mathscr{S}(\Sigma,\mu)$. The solutions for the complex structure, given in Section~\ref{sec:natcomplex}, depend on  the spatial part of the Klein-Gordon equation. Hence, in the following we construct Sobolev spaces w.r.t.\ to the aforementioned operator to have automorphisms w.r.t.\ these  spaces. To make this idea more precise we introduce in this section the notion of weighted Laplace operators \cite[Chapter~3.6, Definition~3.17]{AG1}.
\begin{definition}
	The  respective Laplace-Beltrami operator on the weighted manifold $(\Sigma, \textbf{h},\nu)$ is called the  \emph{weighted Laplace-Beltrami operator} and it is denoted by $\Delta_{ \nu}$ and given by 
	$$\Delta_{ \nu}=\frac{1}{\rho \sqrt{|h|}}\partial_i(\rho\sqrt{|h|}h^{ij}\partial_j).$$  
\end{definition}

Next, we  define a new metric $\tilde{\textbf{h}}$ and measure ${\mu}$ by $$\tilde{\textbf{h}}=N^{-2}\, {\textbf{h}},\qquad \mathrm{and} \qquad d{\mu}=N^{-2}\, d{\nu}=N^{-1}\, d{\Sigma},$$ and we have  the following result (see  a special case of  \cite[Theorem 4.1]{MuOe:selfadjoint}). 
\begin{theorem}\label{mt} 
	Let the Riemannian manifold $(\Sigma, \tilde{\mathbf{h}})$ be   complete and let  the potential   $V\in L^2_{loc}(\Sigma, {\mu})$ be such that it  can be written as $V = V_+ + V_-$, where $V_+\in L^2_{loc}(\Sigma,  {\mu})\geq 0$ and $V_-\in L^2_{loc}(\Sigma,  {\mu})\leq 0$
	point-wise.  Then, the operator
	$w^2$ (from Equation~\ref{op}),    has the form, 
	$$
	w^2=-\tilde{\Delta}_{   {\mu}}+N^2\,V.
	$$
 If the potential is 	 strictly positive, i.e. $V>\epsilon$ for some $\epsilon>0$, then    
	the operator $w^2$	is   essentially self-adjoint   on $C_0^{\infty}(\Sigma)\subset L^{2}(\Sigma,\mu)$ and its closure is strictly positive and invertible.
\end{theorem}

Since the operator $w^2$ is essentially self-adjoint  the domain of the respective maximal operator  i.e.,
$$W_V^2(\Sigma, \mu):=\mathcal{D}(w^2_{max})=\{\Psi\in L^2(\Sigma, \mu):w^2 \Psi\in L^2(\Sigma, \mu)\},$$
is equivalent to the domain obtained by the closure of $w^2$ in $L^2(\Sigma, \mu)$  from the initial domain $C_0^{\infty}(\Sigma)$. Since the closure of the operator $w^2$ which we denote by \begin{equation}\label{hv}
H_V:=\overline{w^2},
\end{equation}    is a non-negative definite self-adjoint operator it follows   (see \cite[Section~VIII.6]{RS2},  \cite[Theorem~X.23]{RS2} and \cite[A.13]{AG1}) that the domain of $H_V$ is a Hilbert space with the scalar product given by,
\begin{equation}
\langle \Psi, \Phi \rangle_{H_V}=\langle H_V\Psi, H_V\Phi \rangle_{L^{2}(\Sigma,\,\mu)}
+\langle \Psi, \Phi \rangle_{L^{2}(\Sigma,\,\mu)} .
\end{equation}
Hence, we have the norm 
\begin{equation}\label{norm}
\Vert \Psi\Vert_{H_V}^2=\Vert H_V\Psi\Vert^2_{L^{2}(\Sigma,\,\mu)}+\Vert\Psi\Vert^2_{L^{2}(\Sigma,\,\mu)}.
\end{equation}
\begin{lemma}\label{noreq}
	The norms  given by $\Vert\cdot\Vert_{H_V}^2$  and by
	\begin{equation}\label{norm1}
	\Vert \Psi\Vert_{W_V^2(\Sigma, {\mu})}^2=\Vert H_V\Psi\Vert^2_{L^{2}(\Sigma,\,\mu)},
	\end{equation} 
	are equivalent.  Moreover, the following inclusions hold $$C_0^{\infty}(\Sigma )\subset W_V^2(\Sigma, \mu)\subset L^2(\Sigma, \mu),$$ and $C_0^{\infty}(\Sigma )$ is dense in $W_V^2(\Sigma, \mu)$.
\end{lemma}
\begin{proof}
	See Appendix~\ref{l42}.
\end{proof}
Next we define the space $W_V^{2s}$ for all real\footnote{The operator $H_V$ with $V\in C^{\infty}$ is an elliptic differential operator and thus enjoys the locality property (i.e.\ it maps $C^{\infty}_0\rightarrow C^{\infty}_0$, see \cite[Chapter~1, Equation~(2.22)]{shand}) and hence we can use  \cite[Lemma~2.1]{ch1} to prove that any integer power of $H_V$ is essentially self-adjoint.}  $s>0$ in an analogous fashion as we did for the case $s=1$. That is, we first define the Hilbert space  
\begin{equation}  \label{ssv}
W_V^{2s}(\Sigma, \mu): =\{\Psi\in L^2(\Sigma, \mu): H_V^{s} \Psi\in L^2(\Sigma, \mu)\} .
\end{equation} 
The domain of the operator $H_V^{s}$ (which is equal to the domain of the operator $(H_V+\rho)^{s}$, \cite[Exercise~4.26]{AG1}) is equivalent to the closure of the operator $H_V^{s}$ in $L^2$ from the initial domain $C_0^{\infty}$ \cite[Chapter~7.1]{AG1}. Another way to see this is using the spectral theorem \cite[Theorem~VIII.6]{RS1}. In particular, we know \cite{MuOe:selfadjoint} that the operator $H_V$ is a   unique strictly  positive   self-adjoint operator. Thus, any   power of the respective operator is  self-adjoint as well. Thus, we use the same arguments to conclude that  the domain of $H_V^{s}$ is the Hilbert space $W_V^{2s}$. This space is equal  to the domain of the closure of the operator which is (due to self-adjointness) the collection of all vectors that are generated by the closure of the operator $H_V^{s}$ in $L^2$ from the initial domain $C_0^{\infty}$. The Hilbert space $W_V^{2s}$ has the scalar product,
\begin{equation}
\langle \Psi, \Phi \rangle_{H_V^{s}(\Sigma, {\mu})}^{2}=\langle H_V^{s}\Psi, H_V^{s}\Phi \rangle_{L^{2}(\Sigma,\,\mu)}
+\langle \Psi, \Phi \rangle_{L^{2}(\Sigma,\,\mu)},
\end{equation}
and hence we have the norm 
\begin{equation}\label{norm2}
\Vert \Psi\Vert_{H_V^{s}(\Sigma, {\mu})}^2 =\Vert H_V^{s}\Psi\Vert^2_{L^{2}(\Sigma,\,\mu)}+\Vert\Psi\Vert^2_{L^{2}(\Sigma,\,\mu)}.
\end{equation}
Note that the space $W_V^{0}(\Sigma, \mu)$ is simply the Hilbert space $L^2(\Sigma, \mu)$.

\begin{lemma}\label{l43a} 
	For the  Hilbert space $W^{2s}_V$ the following inclusions
	$$C_0^{\infty}(\Sigma )\subset W_V^{2s}(\Sigma, \mu)\subset W_V^{2k}(\Sigma, \mu)\subset L^2(\Sigma, \mu),$$
	hold if $s>k$.
\end{lemma}
\begin{proof}
	See Appendix~\ref{pt43a}.
\end{proof}

\begin{theorem}\label{tiso}
	The operator $H_V$ is a linear continuous map $W_V^{2s}(\Sigma, {\mu})\rightarrow W_V^{2s-2}(\Sigma, {\mu})$ and its inverse is a linear continuous map $W_V^{2s-2}(\Sigma, {\mu})\rightarrow W_V^{2s}(\Sigma, {\mu})$. Hence, the operator $H_V$ is an isomorphism between $W_V^{2s}(\Sigma, {\mu})$ and $W_V^{2s-2}(\Sigma, {\mu})$.  
\end{theorem}
\begin{proof}
	To prove continuity the following has to hold, 
	\begin{align*}
	\Vert H_V \Psi\Vert_{W_V^{2s-2}(\Sigma, {\mu})}\leq C 	\Vert  \Psi\Vert_{W_V^{2s}(\Sigma, {\mu})}.
	\end{align*} 
	This inequality follows from the definition of the Hilbert space $W_V^{2s}(\Sigma, {\mu})$. The continuity of the inverse are   proven by Inequality~(\ref{res}) and the inclusions given in Lemma~\ref{l43a}.   
\end{proof}
\begin{corollary}\label{ciso}
	Since the operator $H_V$ is an invertible operator of positive order (order $2$) and its resolvent satisfies Inequality~(\ref{res}) its complex powers, i.e.\ $H_V^l$ for $l\in\bC$ exist and the operator  $H_V^l$ is a topological linear isomorphism between $W_V^{2s}(\Sigma, {\mu})$ and $W_V^{2s-2Re(l)}(\Sigma, {\mu})$.   
\end{corollary}
\begin{proof}
	For a proof of this corollary see \cite{Seeley:1967ea,MCI} and references therein.
\end{proof}

Using these results we can define the expression 
\begin{align}
\langle \langle \Psi_1,\Psi_2\rangle, \langle \Phi_1,\Phi_2\rangle\rangle_
{W_V^{2s+1},W_V^{2s} }=\langle \Psi_1,\Phi_1\rangle_{W_V^{2s+1}}+\langle \Psi_2,\Phi_2\rangle_{W_V^{2s}(\Sigma, {\mu})},
\end{align}
where $\Psi_1,\Phi_1\in W_V^{2s+1}(\Sigma, {\mu})$ and $\Psi_2,\Phi_2\in W_V^{2s}(\Sigma, {\mu})$. This is a scalar product in  the space $W_V^{2s+1}(\Sigma, {\mu})\times W_V^{2s}(\Sigma, {\mu})$. To prove that the complex structure is an isomorphism on the Hilbert space $W_V^{2s+1}(\Sigma, {\mu})\times W_V^{2s}(\Sigma, {\mu})$   we have to prove   
\begin{align}
\Vert J_{Y}\Psi\Vert_{W_V^{2s+1}(\Sigma, {\mu})\times W_V^{2s}(\Sigma, {\mu})} \leq C   \Vert  \Psi\Vert_{W_V^{2s+1}(\Sigma, {\mu})\times W_V^{2s}(\Sigma, {\mu})} ,
\end{align}
for the vector $\Psi\in W_V^{2s+1}(\Sigma, {\mu})\times W_V^{2s}(\Sigma, {\mu})$ and a real positive constant $C$.  
If  the potential of the operator $H_{V}$ is equal to zero or a smooth function with compact support, the constructed Hilbert space $W_V^{2s}(\Sigma, {\mu})$ reduces to the Sobolev space defined for the weighted Laplace operator, see \cite[Chapter~4.2, Chapter~7.1]{AG1}. 

\subsection{Similar Complex Structures}
Let us assume that two complex structures, $J:\mathscr{S}( {\Sigma}, {\mu})\rightarrow \mathscr{S}( {\Sigma}, {\mu})$ and $\overline{J}:\mathscr{S}(  {\Sigma},\overline{\mu})\rightarrow \mathscr{S}(  {\Sigma}, \overline{\mu})$, are similar, i.e., there exists an invertible matrix $X:\mathscr{S}( {\Sigma}, {\mu})\rightarrow \mathscr{S}(  {\Sigma}, \overline{\mu})$ such that
\begin{align*}
J=X^{-1}\overline{J} X.
\end{align*}
Moreover, let us assume that  one of those structures is   anti-self-adjoint. Does that imply   anti-self-adjointness of  the similar structure   as well? The answer is the following lemma.
\begin{lemma}\label{lscs2}
	Let two complex structures, $J:\mathscr{S}( {\Sigma}, {\mu})\rightarrow \mathscr{S}( {\Sigma}, {\mu})$ and $\overline{J}:\mathscr{S}(  {\Sigma},\overline{\mu})\rightarrow \mathscr{S}(  {\Sigma}, \overline{\mu})$,  be similar, i.e., $J=X^{-1}\overline{J}X$ such that the operator valued matrix $X:\mathscr{S}( {\Sigma}, {\mu})\rightarrow \mathscr{S}( {\Sigma}, \overline{\mu})$ fulfills,
	\begin{align*}
	\epsilon \,X^{\dagger}\,\epsilon^T=  \pm  (X^{-1}).
	\end{align*} 
	Moreover, let the complex structure $\overline{J}$ be anti-self-adjoint on $\mathscr{H}_{\overline{J}}$. Then, the complex structure $J$ is anti-self-adjoint on $\mathscr{H}_{{J}}$.
\end{lemma}  \begin{proof}
	See Appendix~\ref{l43}.
\end{proof}

%% file: natcomplex.tex
\section{Natural Complex Structures}

\label{sec:natcomplex}

In the present section we use the machinery developed in the previous sections to study conserved complex structures for certain classes of globally hyperbolic spacetimes. Some of the obtained structures are novel, thus permitting a quantization in cases where none was previously known. In other cases of previously defined or even well-known complex structures our treatment is more general and rigorous (compare Section~\ref{sec:domains}).

In general there is an infinite number of conserved complex structures that can be defined on a given globally hyperbolic spacetime. For example, applying a diffeomorphism will in general lead to a new complex structure. If a timelike Killing vector field exists it can be used to select a complex structure (see for  example \cite{K78}  or \cite{AsMa:qfieldscurved},  \cite{K79}). However,  generically such a vector field will not exist. We thus refrain here from attempting to introduce any precise criterion for choosing a complex structure.
Instead, we consider proposals that appear simple in terms of the mathematical structures that we use to describe the problem, including our choice of coordinates inherent in the decomposition~(\ref{mgh}). We   use the term \emph{natural complex structure} to denote the complex structures that we obtain as solutions to the Conservation Equations.\footnote{This is the same use of language as that of the authors of \cite[Page~3]{AsMa:qfieldscurved}.}

We emphasize that our treatment is focussed on mathematical coherence and consistency. In particular the word ``natural'' should be understood with respect to the mathematical structures, not in the sense that physical criteria as for example coming from a detector model would lead to the solutions we exhibit. The latter considerations are obviously important, but outside the scope of the present article.

\subsection{The Static Case}

We start by considering static spacetimes. While a suitable complex structure was rigorously given in \cite{K78} (see also a less rigorous treatment in \cite{AsMa:qfieldscurved}) we show how it fits into the present framework. Moreover, we generalized  the original theorem (\cite[Theorem 7.2]{K78})  of essential self-adjointness of the operator $w^2$ (see \cite[Theorem~4.1]{MuOe:selfadjoint}, Theorem~\ref{mt}) to all globally hyperbolic stationary spacetimes (with metric of the form Equation~\ref{mgh}).  Before moving forward we recall the precise notion of a static spacetime \cite[Chapter~6, Page~128]{FU}.

\begin{definition}
A metric is \emph{static} if, in appropriate coordinates with $t\equiv x^0$ time-like,
		\begin{enumerate}
			\item $g_{\mu\nu}(\vec{x})$ is independent of $t$, and
			\item $g_{0j}(\vec{x})=0$ for $j=1,\cdots,n$ .
		\end{enumerate} 
If the first condition holds but not necessarily the second, then the metric or geometry is called \emph{stationary}. 
\end{definition}

\begin{remark} \label{remcurv}
	In this section we use throughout the notation $\underline{\mathbf{h}}$ to distinguish a purely space dependent spatial metric from a general one,  ${\mathbf{h}}$. Next, we  recall the explicit form of the scalar curvature for all the metrics in this section. The scalar curvature quantity $\frac{1}{6} N^2\, R$ of the metric $g=- N^{2}(t)\,dt^2+a^{2}(t)	\underline{h}_{ij}(\vec{x})\,d\vec{x}^2$ is given by,
	\begin{equation*}\label{csfrw}
	\frac{1}{6} N^2\, R = a^{-2}	(\partial_t  a)^2-a^{-1}	N^{-1}(\partial_t  a)		(\partial_t  N)+ a^{-1}	\, \partial_t^2  a +\frac{1}{6}a^{-2} N^2\,\underline{R} .
	\end{equation*}
	Here, we denote the scalar curvature of the spatial part with $\underline{R}=	\underline{h}^{ij}\underline{R}_{ij}$, where the underline of the Ricci tensor indicates that it is the curvature tensor for the spatial metric $\mathbf{\underline{h}}$.
\end{remark} 
\begin{remark}  
	In the following subsections we look at spacetimes where the function $f$ (see Equation~(\ref{eq:kgt})) is independent of the spatial components. Hence, the  commutator   
	$$[Z,Y^2]=0,$$
	is equal to zero, 
	and therefore the operator $Z$ commutes with $Y$.\footnote{Since the operator $Y^2$ is a bounded self-adjoint operator its unique square root $Y$ commutes  with all operators that commute with the respective operator $Y^2$ \cite[Chapter~VII, Theorem~2]{RNF}. For a generalization of this important theorem to the unbounded case see \cite{sr3}.\label{rem}} Next, we  use  Conservation equation~(\ref{me4}) and we have for $Z$,
	$$Z=\frac{1}{2}(\partial_tY-f\,Y).$$
	Hence, in all the following spacetimes the Conservation equations reduce to solving a differential equation for the operator $Y$. By the simple equation $A=-Z$ a solution for $Y$ is sufficient to obtain $A$ and therefore the complete complex structure.  
\end{remark} 

Equipped with these definitions we obtain our first theorem.

\begin{theorem}\label{t0} 
	Let a static globally hyperbolic spacetime be given by a metric  $\mathbf{g}$ of the form 
	\begin{align}  
	\mathbf{g}=\begin{pmatrix} 
	-N^2(\vec{x})&  \vec{0}^{\,T}  \\ \vec{0}    &   	\underline{h}_{ij}(\vec{x})
	\end{pmatrix},
	\end{align} 
	and  the  Klein Gordon equation be  
	$$\left(\square_{\mathbf{g}}- \xi R -m^2\right)\phi=0,$$ 
	where $R$ denotes the scalar curvature. If the scalar curvature fulfills the condition $ \xi R>-m^2+\epsilon$ for some $\epsilon>0$, then we have the following solution for the Conservation equations\footnote{In the following we  denote the closure of the operator $w^2$ by the   symbol $H_V$.} 
	\begin{align*}
	Y= H_V^{-1/2}  , 
	\end{align*}  
	where the self-adjoint operator $H_V$ (see Equation~(\ref{hv})) is  
	$$H_V =-\frac{N}{\sqrt{|h|}}\partial_i(
	\sqrt{|h|} N h^{ij}\partial_j
	)+N^2m^{2}+ \xi\,N^2\, R .$$  
	The complex structure (see Equation~(\ref{csm})) defined by the solution given by $Y$ is an anti-self-adjoint operator and an automorphism on the Hilbert space $W_V^{2s+1}(\Sigma, {\mu})\times W_V^{2s}(\Sigma, {\mu})$. This space  is a product of Sobolev spaces defined by powers of the operator  $H_V$. 
\end{theorem} 
\begin{proof} 
	Since $f=0$ 
	    Conservation equation~(\ref{me3}) reduces to 
	\begin{align*}  0=& 
	Y\,\partial_t^2 Y
	-\frac{1}{2} (\partial_tY)^2
	+2w^2\,Y^2-2.
	\end{align*}
	To solve this non-linear and non-homogeneous differential equation we take the time-derivative thereof   and we get, 
	\begin{align*}  
	\partial_t^3 Y +4w^2\partial_tY&=0 .
	\end{align*} 
	Since the operator $w^2$ does not depend on the time, we solve the second differential equation by standard methods, 
	\begin{align*}  
	Y(t)=c_1(w)^{-1}e^{2iwt}+c_2(w)^{-1}e^{-2iwt}+c_3.
	\end{align*} 
	After inserting the solution in the  first differential equation and by demanding the operator $Y$ to be symmetric and positive (due to the inner product, see Equations~(\ref{sop})) we obtain the following relations for the constants, 
	$$ c_3=+\sqrt{1+4c_1c_2}\,(w^{-1}) ,\qquad c_1=c^*_2.$$ 
	For $c_1=0$, we prove   anti-self-adjointness of the  complex structure   by proving that Assumption~\ref{ass1} (see Theorem~\ref{mt} and Proposition \ref{cssa}) is satisfied. That  Assumption~\ref{ass1} holds follows from the fact that the operator $H_V^{1/2}$ is self-adjoint (see Theorem~\ref{mt}). Thus, any power thereof is self-adjoint as well (due to the spectral theorem) and hence the operators $Y$ and $Y^{-1}$ are self-adjoint.  Next, we prove the property of the complex structure being an automorphism by using the construction of Section~\ref{csauto}, i.e.\ 
	\begin{align}\nonumber
	\Vert J_{ {Y}}\Psi\Vert_{W_V^{2s+1}(\Sigma, {\mu})\times W_V^{2s}(\Sigma, {\mu})}^2& =
	\Vert H_V^{-1/2}\Psi_2\Vert_{W_V^{2s+1}(\Sigma, {\mu})}^2 + \Vert  H_V^{1/2}\Psi_1\Vert_{W_V^{2s}(\Sigma, {\mu})} \\&\nonumber\leq
	c_4 \Vert \Psi_2\Vert_{W_V^{2s}(\Sigma, {\mu})}^2 + c_5\Vert \Psi_1\Vert_{W_V^{2s+1}(\Sigma, {\mu})}\\&\label{ineqcs}
	\leq
	c\Vert  \Psi\Vert_{W_V^{2s+1}(\Sigma, {\mu})\times W_V^{2s}(\Sigma, {\mu})}^2,
	\end{align}  
	where we used the fact that the vector $\Psi=(\Psi_1,\Psi_2)\in W_V^{2s+1}(\Sigma, {\mu})\times W_V^{2s}(\Sigma, {\mu})$ and we used Theorem~\ref{tiso} and Corollary~\ref{ciso}.   
\end{proof}

There are several examples of static spacetimes that are interesting in their own right. One  example is the Schwarzschild spacetime \cite{RF}. Another interesting example is    the case of the ultra-static spacetimes (see for instance \cite{US}). Ultra-static spacetimes are globally hyperbolic if and only if the manifold $(\Sigma,\mathbf{h})$ is complete \cite{K78}.  Hence, essential self-adjointness  (Theorem~\ref{mt},  \cite[Theorem~4.1]{MuOe:selfadjoint}) of the operator $w^2$ (see Equation~(\ref{op})) holds for all globally hyperbolic ultra-static spacetimes  and all spacetimes that are respectively conformally isometric. In particular, that includes static spacetimes. Thus,  for the case of static globally hyperbolic spacetimes  $(M,\mathbf{g})$   there always exists a foliation ($M\cong\R\times\Sigma$) such that the resulting Cauchy surface $\Sigma$ and the induced metric $\tilde{\mathbf{h}}$, considered as a Riemannian manifold $(\Sigma,\tilde{\mathbf{h}})$  is a complete metric space. In turn this means that for all such spacetimes (and conformal transformations thereof) that fulfill the condition $\xi R>-m^2+\epsilon$ for some $\epsilon>0$, we  gave a well-defined complex structure.\footnote{In \cite{K78} the author constructed a well-defined complex structure for static spacetimes under some boundedness requirements on the  coefficients of the metric. However, in this work (and in \cite{MuOe:selfadjoint}) we remove these restrictions.}
  To compare our results with \cite{AsMa:qfieldscurved}  and  \cite{K78} see \cite[Section~7.2]{CCQ:schroefock} and \cite{MCI}.

\subsection{The Expanding Spacetimes}

In this section we focus on the case where the spatial part of the metric does not depend on time, but the    $g_{00}$-component of the metric does. Such globally hyperbolic spacetimes can emerge from   stationary spacetimes that   differ  from the representation of the metric (see Equation~(\ref{mgh})) by the existence of a so-called shift vector $\vec{N}$, i.e.,
$$ds^2=-(N^2-N_iN^i)dt^2+2N_idtdx^i+h_{ij}dx^i\,dx^j,$$ 
where all functions, vectors and matrices (i.e.\ $N,\,\vec{N},\,\mathbf{h}$) depend on the spatial coordinates.  Yet, all stationary   globally hyperbolic  spacetimes  admit a form such as Equation~(\ref{mgh}). This follows from  Theorem~\ref{T1}. The coordinate transformation that performs this change    will, however, induce time dependencies in the respective coefficients of the metric. Hence, given a stationary globally hyperbolic  spacetime (where the metric does not depend on the time) and transforming it into the form studied here (see  Equation~(\ref{mgh})) will induce time dependencies. In some cases the  time dependencies are of the form of the metric studied in this section.

\begin{theorem}\label{t1}  
	Let the  metric  $\mathbf{g}$ and  the  Klein-Gordon equation be given by 
	\begin{align} \label{mest}
	\mathbf{g}=\begin{pmatrix} 
	-N^{2}(t)&  \vec{0}^{\,T}  \\ \vec{0}    &   	\underline{h}_{ij}(\vec{x})
	\end{pmatrix},\qquad \qquad \left(
	\square_{\mathbf{g}}- \xi R -m^2\right)\phi=0,
	\end{align} 
	where $R$ denotes the scalar curvature.  Let the curvature scalar fulfill the positivity condition $\xi R>-m^2+\epsilon$ for some $\epsilon>0$. Then,   a solution  for the Conservation equations   is, 
	\begin{align*}
	Y=H_V^{-1/2},
	\end{align*} 
	where the operator $H_V$ is given by 
	\begin{align*}
	H_V =N^2
	\left(-\Delta_{\underline{h}}+\xi\,{R}+m^2
	\right)  .		\end{align*} 
	The complex structure (see Equation~(\ref{csm})) defined by the solution of $Y$ is an anti self-adjoint operator and an automorphism on the Hilbert space $W_V^{2s+1}(\Sigma, {\mu})\times W_V^{2s}(\Sigma, {\mu})$. This space is a product of Sobolev spaces defined by powers of the operator  $H_V$. 
\end{theorem} 
\begin{proof} 
	See Appendix~\ref{pt1}. 
\end{proof}

\subsection{FRW-type Spacetimes and Conformal Classes}

In this section we investigate FRW-type cosmologies\footnote{Homogeneous and isotropic cases, conventionally called FRW spacetimes (see Metric (\ref{frwmetr})), imply constant curvature, \cite[Chapter 5.1]{Wal:gr}. In the following sections we call FRW-type spacetimes,  spacetimes with  metric of the form as in Equation (\ref{frwmetr}) where the spatial curvature is not necessarily constant.}. In these  
models the Klein-Gordon equation can be solved by separation of variables. However, for  some of the   models, e.g.\ QFT in FRW spacetime with $m\neq0$ and $\xi\neq1/6$, there exists no natural choice of basis for the space of solutions analogous to the positive- and negative-frequency exponentials in 
the conformally trivial case (\cite[Chapter 7]{FU}). Nevertheless,   
in the following we supply examples  where we construct complex structures for these spacetimes on each Cauchy surface such that there exists a non-trivial splitting 
into positive- and negative-frequencies (on each hypersurface).
All the above-mentioned spacetimes are such that    the function  
$f=-N^{-1} ( \partial_t N ) +  (\sqrt{|h|})^{-1}  \,(\partial_t\sqrt{|h|})$ (see Equation~(\ref{eq:kgt})) is independent of the spatial directions and hence the equations  that define a conserved complex structure (see Proposition~\ref{p2}) simplify.   Friedmann-Robertson-Walker type spacetimes are represented by a metric   of the form
\begin{align}\label{frwmetr}
\mathbf{g}=\begin{pmatrix} 
-N^{2}(t)&  \vec{0}^{\,T}  \\ \vec{0}    &   a^2(t)	\underline{h}_{ij}(\vec{x})
\end{pmatrix}.  
\end{align}
These spacetimes     have applications in astrophysics, in the description of the universe and in astro-particle physics (see the excellent reference \cite{THC}).   In this section the metric (Equation~(\ref{frwmetr})) is conformally isometric  to the previously studied metric  (see Equation~(\ref{mest})). Instead of solving the Conservation equations   for the FRW-case explicitly, we first construct transformations that take us from
the former objects $(\mathbf{g}, J)$ to the latter objects $(\overline{\mathbf{{g}}}, \overline{J})$. By using these constructed transformations   we obtain   solutions for the case at hand.  We denote a  conformal transformation of the metric by
\begin{align}\label{ct}
\mathbf{g} \mapsto  \overline{\mathbf{{g}}}= \Lambda^2\,  \mathbf{g}.
\end{align} 
Using this transformation one has for the massless Klein-Gordon equation (in four dimensions) with the potential being $V=\frac{1}{6}R$ the following (see \cite[Equation~(3.5)]{BD}),
\begin{align*}\left(
\square_{\mathbf{g}}- \frac{1}{6}R\right)\phi=0\longrightarrow  \left(
\square_{\overline{\mathbf{g}}}- \frac{1}{6}\overline{R}\right)\overline{\phi}= \Lambda^{-3} \left(
\square_{\mathbf{g}}- \frac{1}{6}R\right)\phi=0,
\end{align*} 
where   
\begin{align}\label{ctf}
\overline{\phi}=\Lambda^{-1}\phi,
\end{align} 
and $\overline{R}$ is the scalar curvature of the metric $\overline{\mathbf{g}}$.  The following lemmas give  explicit transformations  for the   complex structures corresponding to the metrics that are conformally isometric. 

\begin{lemma}\label{tcts} 
	Let the metric $\mathbf{g}$ be given by  
	\begin{align*}  
	\mathbf{g}=\begin{pmatrix} 
	-N^{2}(t)&  \vec{0}^{\,T}  \\ \vec{0}    &   a^2(t)	\underline{h}_{ij}(\vec{x})
	\end{pmatrix}  ,
	\end{align*}   
	and let the conformal transformation be explicitly given by $\Lambda^2=a^{-2}$, i.e.  
	\begin{align*}  
	\mathbf{g}  \longrightarrow 	 \overline{\mathbf{g}}=\begin{pmatrix} 
	-N^{2}(t)a^{-2}(t)&  \vec{0}^{\,T}  \\ \vec{0}    & 	\underline{h}_{ij}(\vec{x})
	\end{pmatrix}=:\begin{pmatrix} 
	-	\overline{N}^{2}(t)&  \vec{0}^{\,T}  \\ \vec{0}    &   	\underline{h}_{ij}(\vec{x})
	\end{pmatrix} .
	\end{align*}   
	Then, the conformal transformation (see Equations~(\ref{ct}), (\ref{ctf})) acts on vectors of the canonical phase space as follows, 
	\begin{align}\label{xmatr}
	\begin{pmatrix}\varphi\\\pi\end{pmatrix}\longrightarrow
	\begin{pmatrix}\overline{\varphi}\\\overline{\pi}\end{pmatrix}&= 
	\begin{pmatrix} 
	a &  0 \\ 	\overline{N}^{-1} \,\sqrt{	\underline{h}}\,  (\partial_{t}a)    &   a^{-1}
	\end{pmatrix} 
	\begin{pmatrix}{\varphi}\\{\pi}\end{pmatrix} =: 
	X
	\begin{pmatrix}{\varphi}\\{\pi}\end{pmatrix}.
	\end{align} Moreover,   the conformal transformation    leaves  the  symplectic structure invariant,
	\begin{align*} 
	\Omega  ( \Phi_1, \Phi_2) =\Omega  ( \overline{\Phi}_1, \overline{\Phi}_2) .
	\end{align*} 
\end{lemma} \begin{proof}
See Appendix~\ref{ptcts}.
\end{proof}

\begin{lemma}\label{tctsb} Let us denote the   complex structure w.r.t.\ to the    metric  ${\mathbf{g}}$ by ${J}$ and the complex structure of the conformally related metric $\overline{\mathbf{g}}$ by $\overline{J}$. Then, the  complex structures ${J}$ and $\overline{J}$  are similar, i.e.,
	\begin{align*} 
	J=X^{-1}\overline{J}X,
	\end{align*} 
	where the matrix $X$ is given in the previous Lemma (see Equation~(\ref{xmatr})). Moreover, the complex structures w.r.t.\ the measure space $(\Sigma, {\mu})$ acting on scalar functions are given by,  
	\begin{align*} 
	J_{ {Y}}=X_{T }^{-1}\overline{J}_{Y}X_{T } ,
	\end{align*} 
	with the operator-valued matrix $X_{T }:=T \,X \,T ^{-1}$, where the  transformation matrix $T$ is given  in Equation~(\ref{tmat}). 
\end{lemma}
\begin{proof}
See Appendix~\ref{ptctsb}.
\end{proof}
Next, we use these lemmas to obtain the solutions for the Conservation equations for  the spacetime at hand.

\begin{theorem}\label{t4}
Let the  metric  $\mathbf{g}$ and  the Klein Gordon equation read 
	\begin{align}\label{met4}
	\mathbf{g}=\begin{pmatrix} 
	-N^{2}(t)&  \vec{0}^{\,T}  \\ \vec{0}    &   a^2(t)	\underline{h}_{ij}(\vec{x})
	\end{pmatrix},\qquad \qquad \left(
	\square_{\mathbf{g}}- \frac{1}{6}R \right)\phi=0,
	\end{align} 
	where $R$ denotes the scalar curvature.  Let the scalar curvature of the spatial metric $\underline{\mathbf{h}}$ be denoted by $\underline{R}$. If the spatial scalar curvature is positive, i.e.\  $\underline{R} >\epsilon$ for some $\epsilon>0$,   then we have the following solution  for the Conservation equations  
	\begin{align*}
	Y=H_V^{-1/2}.
	\end{align*}  
	The complex structure defined by the solutions, denoted as $J_{Y}$, is similar to the complex structure $\overline{J}_{Y}$ obtained in Theorem~\ref{t1} with similarity matrix,  
	\begin{align*}X=
	\begin{pmatrix} 
	a &  0 \\ 	{N}^{-1}  \,\sqrt{	\underline{h}}\, a (\partial_{t}a)    &  a^{-1}
	\end{pmatrix} ,
	\end{align*} 
	where the operator  $H_V$ is given by  
	$$
        H_V=\left(-
	\Delta_{\underline{h}}+\frac{1}{6}\,\underline{R} 
	\right) \,{N}^{2}\,{a}^{-2}.
        $$
	Moreover, the complex structure $J_{Y}$ is anti-self-adjoint and an automorphism    on the Hilbert space $W_V^{2s+1}(\Sigma, {\mu})\times W_V^{2s}(\Sigma, {\mu})$. This space is a product of Sobolev spaces defined by powers of  the operator  $H_V$.
\end{theorem}    
\begin{proof} 
	See Appendix~\ref{t54}. 
\end{proof}

\subsection{QFT in FRW-type Spacetimes with Mass Term}

In this subsection  we find solutions of the complex structure for different FRW-models of  the massive Klein-Gordon equation. 
\begin{theorem}\label{t6} 
	Let the metric be $\mathbf{g}=- N^{2}(t)\,dt^2+a^{2}(t)	\underline{h}_{ij}(\vec{x})\,dx^idx^j$ and   let the Klein-Gordon equation be given with the time dependent mass term $M^2(t)=a^{-2}(t)m^2$, i.e.\  
	\begin{align*}\left(
	\square_{\mathbf{g}}- \frac{1}{6}R-M^2(t)\right)\phi=0.
	\end{align*}Let the scalar curvature of the spatial metric $\underline{\mathbf{h}}$ be denoted by $\underline{R}$ and let it fulfill the positivity condition $\frac{1}{6}\, \underline{R}>-m^2+\epsilon$ for some $\epsilon>0$. Then,   a solution  for the Conservation equations  is,  
	\begin{align*}	
	Y=H_V^{-1/2}, 
	\end{align*}  
	where the operator $H_V$ is  
	\begin{align*}
	H_V =
	\left(-\Delta_{\underline{h}}+\frac{1}{6}\,\underline{R}+m^2
	\right) N^2a^{-2}. 		\end{align*} 
	Moreover, the complex structure $J_{Y}$ (see Equation~(\ref{csm})) is anti-self-adjoint and an automorphism    on the Hilbert space $W_V^{2s+1}(\Sigma, {\mu})\times W_V^{2s}(\Sigma, {\mu})$. This space is a product of Sobolev spaces defined by powers of  the operator  $H_V$.
\end{theorem} 
\begin{proof} 
	See Appendix~\ref{t55}. 
\end{proof}
The physical interpretation of the time-dependent mass term becomes clear by   conformally transforming the Klein-Gordon equation,
\begin{align*}\left(
\square_{\mathbf{g}}- \frac{1}{6}R\right)\phi \longrightarrow  \left(
\square_{\overline{\mathbf{g}}}- \frac{1}{6}\overline{R}\right)\overline{\phi}= \Lambda^{-3} \left(
\square_{\mathbf{g}}- \frac{1}{6}R\right)\phi,
\end{align*} 
where for the Klein-Gordon equation with time dependent mass term we have, 
\begin{align*} \left(
\square_{\overline{\mathbf{g}}}- \frac{1}{6}\overline{R}-m^2\right)\overline{\phi}&= \Lambda^{-3} \left(
\square_{\mathbf{g}}- \frac{1}{6}R\right)\phi-m^2\overline{\phi}\\&=
\Lambda^{-3} \left(
\square_{\mathbf{g}}- \frac{1}{6}R-\Lambda^{2}m^2\right)\phi .
\end{align*}
This equation makes clear (the well-known fact) that a mass term breaks conformal invariance of the Klein-Gordon equation. Conformal invariance  was used in the previous results to get  complex structures for a class of conformally isometric spacetimes. Hence, if we conformally transform the metric $g$, the Klein Gordon equation transforms such that the time dependent mass term becomes independent of time. Therefore, we obtain the usual Klein-Gordon equation with the potential $\frac{1}{6}\overline{R}$. In the following two examples the same interpretation for the solutions holds.
 
The next result resembles well studied examples in \cite{Cortez} and references therein. 
 \begin{theorem}\label{t7}
 	
 	Let the metric be $\mathbf{g}=- a^{2}(t)\,dt^2+a^{2}(t)	\underline{h}_{ij}(\vec{x})\,dx^idx^j$ and   let the Klein-Gordon equation be given with the time dependent mass term $M^2(t)=a^{-2}m^2+(1-6\xi)a^{-3}a''$, i.e.\  
 	\begin{align*}\left(
 	\square_{\mathbf{g}}- \xi R- M^2(t)\right)\phi=0.
 	\end{align*}Let the scalar curvature of the spatial metric $\underline{\mathbf{h}}$ be denoted by $\underline{R}$ and let it   fulfill the positivity condition $\xi\, \underline{R}>-m^2+\epsilon$ for some $\epsilon>0$. Then,   a solution for the Conservation equations  is,  
 	\begin{align*}	
 	Y=H_V^{-1/2},
 	\end{align*}  
 	where the operator $H_V$ is  
 	\begin{align*}
 	H_V = -\Delta_{\underline{h}}+\xi\,\underline{R}+m^2
 	   . 		\end{align*}  
 Moreover, the complex structure $J_{Y}$ (see Equation~(\ref{csm})) is anti-self-adjoint and an automorphism    on the Hilbert space $W_V^{2s+1}(\Sigma, {\mu})\times W_V^{2s}(\Sigma, {\mu})$. This space is a product of Sobolev spaces defined by powers of  the operator  $H_V$.
\end{theorem}
\begin{proof} 
	See Appendix~\ref{t66}. 
\end{proof}  
 For the following examples we insert concrete functions for the scale factor, namely those corresponding to the scenarios of a radiation-dominated universe ($a=t^{1/2}$) and to  inflation ($a=e^{Ht}$), see \cite[Chapter 5 and 6]{Wal:gr}, \cite[Chapter 7]{FU} and in a QFT context see also \cite{ad1} and \cite{ad2}. The inflation example can also be considered as flat slicing in de~Sitter space. 
 \begin{example}\label{lem:11b} 
 	Let the metric be $\mathbf{g}=-   dt^2+a^{2}(t)	\underline{h}_{ij}(\vec{x})\,dx^idx^j$
 	and let the scale factor  be $a=t^{1/2}$ with $t\in\mathbb{R}^{+}$. Moreover, let the Klein-Gordon equation be given by
 	\begin{align*}\left(
 	\square_{\mathbf{g}}- \xi R- m^2 \right)\phi=0.
 	\end{align*}
 	Then, a solution  for the Conservation equations  is, 
 	\begin{align*}
 	Y(t)=c_1 M_{\alpha,1/4}(z(t))^2+c_2 M_{\alpha,1/4}(z(t)) W_{\alpha,1/4}(z(t)) +c_3 W_{\alpha,1/4}(z(t))^2
 	\end{align*} 
 	where we define the variables $\alpha:=-iH_V /(2m)$, $z(t):=2im\,t$.  Moreover,
 	$M_{\alpha,\mathbb{R}}$ and $W_{\alpha,\mathbb{R}}$ are standard Whittaker functions, $c_1,\,c_2,\,c_3$ are real valued constants and the operator $H_V$ is  given by,
 	\begin{align*}
 	H_V  = (-\Delta_{\underline{h}}+\xi\underline{R} )  .
 	\end{align*}  		 
 \end{example}
 \begin{proof}
 	See Appendix~\ref{t88b}. 	 
 \end{proof}  
\begin{example}\label{lem:10} 
	Let the metric be $\mathbf{g}=-   a^{2}(t)dt^2+a^{2}(t)	\underline{h}_{ij}(\vec{x})\,dx^idx^j$. Moreover,   let the Klein-Gordon equation be given by
\begin{align*}\left(
\square_{\mathbf{g}}- \xi R- m^2 \right)\phi=0.
\end{align*}
 
For   scale factor   $a=e^{Ht}$, with $H$ being a parameter and scalar curvature   $\underline{R}$,   fulfilling the positivity condition $\xi \underline{R}>   -(6\xi-1)H^2 + \epsilon$ for some $\epsilon>0$ we have  the  following solution  for the Conservation equations,
\begin{align*}	
Y(t)& =  c_1 \,   \Gamma(1-\alpha)\Gamma(1+\alpha)
J_{\alpha}( {z(t)})\,J_{-\alpha}( {z(t)})\\& +	c_2\, 
4^{- \alpha}  \,\Gamma(1-\alpha)^2\,J_{-\alpha}( {z(t)})^2
+c_3 \,4^{  \alpha}\,  
\Gamma(1+\alpha)^2\,J_{\alpha}( {z(t)})^2,	
\end{align*}  
where  we defined the variables $\alpha:=iH_V^{1/2}/H$ and $-z^2(t):=-{e^{ 2 H t}m^2}/{(   H^2) } $. Furthermore, the symbol	$J_{\nu}(\tau)$ denotes the Bessel function of first order, $c_1,\,c_2,\,c_3$ are real valued constants and the operator $H_V$ is  given by,
\begin{align*}
H_V = (-\Delta_{\underline{h}}+\xi\underline{R}+ (6\xi-1)H^2) . 		
\end{align*}
\end{example}
\begin{proof}
	See Appendix~\ref{t77}. 	 
\end{proof}

\begin{example}\label{lem:11} 
	Let the metric be $\mathbf{g}=-   dt^2+a^{2}(t)	\underline{h}_{ij}(\vec{x})\,dx^idx^j$
	and let  $a=e^{Ht}$, with $H$ being a parameter. Moreover,   let the Klein-Gordon equation be given by
	\begin{align*}\left(
	\square_{\mathbf{g}}- \xi R- m^2 \right)\phi=0.
 	\end{align*}Let the mass of the particle fulfill the positivity condition $m^{2} >- 2\,( 6\xi  -\frac{9}{8}   )H^2 +\epsilon$   for some $\epsilon>0$.  	Then,    a solution  for the Conservation equations  is, 
	\begin{align*}	
	Y(t) & 
	=c_1 \,   \Gamma(1-\alpha)\Gamma(1+\alpha)
	J_{\alpha}( {z}(t))\,J_{-\alpha}( {z}(t))\\& +	c_2\,  
	4^{- \alpha}  \,\Gamma(1-\alpha)^2\,J_{-\alpha}( {z}(t))^2
	+c_3 \,4^{  \alpha}\, 
	\Gamma(1+\alpha)^2\,J_{\alpha}( {z}(t))^2 
	\end{align*} 
		where  we defined    the variables $M^2:=m^{2}+2\,( 6\xi  -\frac{9}{8}   )H^2$, $\alpha:=-iM/H$ and $-z^2(t):=-{e^{- 2 H t}H_V }/{(   H^2) }$. Moreover,
	 $J_{\nu}(\tau)$ is a Bessel function of first order, $c_1,\,c_2,\,c_3$ are real valued constants and the operator $H_V$ is  given by,
		\begin{align*}
		H_V  = (-\Delta_{\underline{h}}+\xi\underline{R} )  . 		
		\end{align*}
\end{example}
\begin{proof}
	See Appendix~\ref{t88}. 	 
\end{proof}
 
\begin{example}\label{lem:12}
	Let the metric be $\mathbf{g}=-   a^{2}(t)dt^2+a^{2}(t)	 \underline{h}_{ij}(\vec{x})\,dx^idx^j$. Moreover,   let the Klein-Gordon equation be given by
	\begin{align*}\left(
	\square_{\mathbf{g}}- \xi R- m^2 \right)\phi=0,
	\end{align*}
   with scale factor   $a(t)={r}/{t}$,   $r\in\R^{+}$, $t\in(0,\infty)$ and
   scalar curvature   $\underline{R}$    fulfilling the positivity condition $\xi \underline{R}>     \epsilon$ for some $\epsilon>0$
   we have  the  following solution  for the Conservation equations,
	\begin{align*}	
	Y(t)& = c_1 \,t\, J_{\alpha}(H_V^{1/2} \,t)^2 + c_3  \,t\, Y_{\alpha}(H_V^{1/2}\, t)^2 + c_2 \,t\, J_{\alpha}(H_V^{1/2}\, t) Y_{\alpha}(H_V^{1/2}\,t)
 	\end{align*}  
	where  we defined the variables $\alpha:= \sqrt{9(1-6\xi)/4-(r\,m)^2}$. Furthermore, the symbols	$J_{\nu}(\tau)$ and $Y_{\nu}(\tau)$ denote  the Bessel functions of first and second order respectively, $c_1,\,c_2,\,c_3$ are constants and the operator $H_V$ is  
	\begin{align*}
	H_V = (-\Delta_{\underline{h}}+\xi\underline{R}). 		
	\end{align*}
\end{example}
\begin{proof}
	See Appendix~\ref{t88c}. 	 
\end{proof} 
 From these examples  the physical meaning of the operator $Y$ (in the context of the FRW and de~Sitter spacetimes) becomes clear. First note that the examples discussed in this section can be solved by separation of variables in terms of space and time. In order to focus the discussion let us specialize to a spatially and conformally flat Robertson-Walker spacetime with metric,
$$g=-a^2(t)\,dt^2+a^2(t)\,\delta_{ij}\,dx^i\,dx^j .$$ 
    Then the solutions of the field equation can be expanded in terms of momentum modes \cite[Page 136]{FU} as,   
 \begin{align*}
 \phi(t,x)= (2\pi)^{-3/2}\,a(t)^{-1}\,\int\,d^3\mathbf{k}\,[\phi(\mathbf{k})\,\chi_{\mathbf{k}}(t)\,e^{-i\mathbf{k}\cdot\mathbf{x}}+h.c.],
 \end{align*}
where     $\chi_{\mathbf{k}}(t)$  is the generalization of the momentum plane wave $e^{i\omega_{\mathbf{k}} t}$ in Minkowski spacetime. It is the solution of 
the differential equation,
\begin{align}\label{eq:chi}
\frac{\partial^2 \chi_{\mathbf{k}}}{\partial t^2}+\left(\Omega_{\mathbf{k}}(t)^2+\left(\xi-\frac{1}{6}\right)R\,a^2(t)\right)
\chi_{\mathbf{k}}=0,
\end{align}
where $\Omega_{\mathbf{k}}(t)^2=\mathbf{k}^2+m^2\,a(t)^2$.

Next, we restrict to the particular case of Example~\ref{lem:12}. Then, we can take $\chi_{\mathbf{k}}=t^{1/2}\,H_{\alpha}(\mathbf{k}\,t)$, where $H_{\alpha}(\mathbf{k}\,t)$ is the Hankel function of order $\alpha$. The eigenmodes of the  operator $H_V$ are the momentum eigenmodes that we denote $\Psi_{\mathbf{k}}$. We fix the constant $c_1$ arbitrarily and choose $c_2,\,c_3$  as follows $$c_2=2i\,c_1,\qquad \qquad  c_3=-c_1.$$ Then, the operator $Y$ is     the square of    $\chi_{\mathbf{k}}$ when acting upon eigenfunctions of the operator $H_V$, i.e.,
\begin{align*}
Y(t)\Psi_{\mathbf{k}}=c_1\,\chi_{\mathbf{k}}^2(t)\Psi_{\mathbf{k}}.
\end{align*}
Hence, one can interpret the operator $Y$ as the operator representation of (the square of)  the function  $\chi_{\mathbf{k}}$. In the present case it turns out, moreover, that restricting to minimal coupling ($\xi=0$) the complex structure corresponding to this choice of $Y$ (and $A$)  encodes precisely the well known Bunch-Davis vacuum \cite{ChTa:qftsds}. The present considerations also illustrate one advantage of our approach: We can construct and describe a vacuum (in terms of the operators $Y$ and $A$ or $J$ itself) without prior recurrence to a specific mode expansion. 
\begin{remark}
	Although we demonstrate it for a specific scale factor $a(t)$ the relation ($Y(t)\Psi_{\mathbf{k}}=c_1\,\chi_{\mathbf{k}}^2(t)\Psi_{\mathbf{k}}$) holds for an arbitrary scale factor $a(t)$. This can be seen  in general since  the  differential Equation \eqref{eq:ddiffeq} for $Y$ is equivalent to 	the differential Equation \eqref{eq:chi}   when acting upon eigenfunctions of the operator $H_V$.
\end{remark}
\subsection{Complex Structures and the Gelfand-Dikki Equation}
From all the former (mostly well-known) examples a general thread emerges.  
Namely, if the operator $w^2$ and the function $f$ commute (on some dense domain of $w^2$) we have the differential equation 
\begin{align}  
	Y\,\partial_t^2 Y
	-\frac{1}{2} (\partial_tY)^2
	-Y^2\left(
	\partial_tf +\frac{1}{2}f^2 -
	2  w^2
	\right)=2.
\end{align} 
The solution of $Y$ completely determines, in this case, the form of the complex structure. It is interesting to point out that this is the Gelfand-Dikii equation (see \cite[Chapter 2, Equation~(6)]{geld})\footnote{We are indebted to Christian Schubert for pointing this out.}. This equation is an important part in the Sturm-Liouville theory and it is connected to the Ricatti Equation, see \cite[Footnote 1, Page 97 and Equations (11)-(13), Page 98]{geld} as well.  To see the equivalence to the Gelfand-Dikii equation we apply the former equality to eigenfunctions  $\Psi_k$  of the self-adjoint operator $w^2$ (i.e. $w^2\Psi_k=w^2_k\Psi_k$), define the variables \begin{align}
	2(u(t)+\xi):= 
	\partial_tf +\frac{1}{2}f^2 -
	2  w^2,\qquad \qquad Y=\pm2iG
\end{align}
where $\xi\in\R$ and denote the time-derivative with the prime symbol
\begin{align*}
	-2GG^{''}+(G')^2+4(u(t)+\xi)G^2=1.
\end{align*}
Taking the time-derivative of this equation we obtain (as in our various proofs) 
\begin{align*}
	- G^{'''}+4(u(t)+\xi)G^{'}+2u'(t)G=0.
\end{align*}
This, rather complicated, differential equation is solved by a linear combination 
of $\varphi_1^2$, $\varphi_2^2$ and $\varphi_1\varphi_2$, where $\varphi_1$ and $\varphi_2$ are fundamental solutions of the differential equation, 
\begin{align*}
	-\varphi''+(u(t)+\xi)\varphi=0.
\end{align*} 
Note that in general there exists an asymptotic expansion of the function $G$ given by 
\begin{align*}
	G(x;\xi)=\sum_{l=0}^{\infty}\frac{G_l[u]}{\xi^{l+1/2}}
\end{align*}
as $\xi\rightarrow\infty$. The expansion coefficients $G_l$ are polynomials of the function $u$ and its derivatives and are given by a recurrence relation, see \cite[Equation 9, Page 97]{geld}. \\\\
Therefore, at least recursively one  can  solve the complex structures for a general class of Klein-Gordon theories for which the operator $w^2$ and the function $f$ commute.   

%% file: ackn.tex
\section*{Acknowledgments}
The authors would like to thank Alejandro Corichi for pointing out important references and for various discussions. Moreover, one of the authors (AM) would like to thank  Bernard Kay for   enlightening  discussions on quantum field theory in curved spacetimes and related functional analytical questions. 
This work was partially supported by CONACYT project grant 259258.

%% file: proofs.tex
\section{Proofs}\label{sec:proofs}
   
\addtocontents{toc}{\begingroup\string\c@tocdepth 1\relax}

  \stepcounter{subsection}
\subsection{Proofs of Section~\ref{sec:ghyp}}  \stepcounter{subsubsection}
\subsubsection{Proof of Observation~\ref{p1}}\label{pp1} 
	\begin{proof} In order to see the explicit form of the Klein Gordon Equation we first use the splitting of a globally hyperbolic spacetime provided by Theorem~\ref{T1}, i.e.\   	\begin{align*}  
		\textbf{g}=\begin{pmatrix} 
		-N^{2} &  \vec{0}^{\,T}  \\ \vec{0}    &   	{h}_{ij} 
		\end{pmatrix}, \qquad |g|=N^{2}|h|, \sqrt{|g|}=N\,\sqrt{|h|},
		\end{align*} 
		and then a straightforward calculation yields, 
		\begin{align*} &
		\left(({\sqrt{|g|}})^{-1}\partial_a(\sqrt{|g|}g^{ab}\partial_b)-m^2(x)\right)\phi \\&=
		\left(	(N\,\sqrt{|h|})^{-1}\partial_0(   N\,\sqrt{|h|}    g^{00}    \partial_0 ) + (N\,\sqrt{|h|})^{-1}\partial_i(   N\,\sqrt{|h|}    h^{ij}     \partial_j )-m^2(x)\right)\phi\\&=  
		\left(-	(N^2\,\sqrt{|h|})^{-1}(   -N^{-1} ( \partial_0 N )\,\sqrt{|h|}\partial_0 +    \,(\partial_0\sqrt{|h|})  \partial_0 +   \,\sqrt{|h|} \partial^{2}_t ) - N^{-2}w^2 \right)\phi=0,
		\end{align*}
		which in turn  leads to 
		\begin{align*} 
		&	\left( 	( \,\sqrt{|h|})^{-1}(   -N^{-1} ( \partial_0 N )\,\sqrt{|h|}\partial_0 +    \,(\partial_0\sqrt{|h|})  \partial_0 +   \,\sqrt{|h|} \partial^{2}_t ) + w^2 \right)\phi\\=&
		\left( 	 \partial^{2}_t         -N^{-1} ( \partial_0 N ) \partial_0 +  (\sqrt{|h|})^{-1}  \,(\partial_0\sqrt{|h|})  \partial_0 +     w^2 \right)\phi= 0,
		\end{align*} 
		where the operator $w$ is given by $$w^2=-\frac{N}{\sqrt{|h|}}\partial_i(
		\sqrt{|h|} N h^{ij}\partial_j
		)+N^2m^{2}(x)=
		-N^2(\Delta_{h}-m^{2}(x))-Nh^{ij}\partial_iN\partial_j,$$		
		with $\Delta_h$ being the Laplace-Beltrami operator for the associated $3$-metric $h_{ij}$.
		
	\end{proof}
	\subsection{Proofs of Section~\ref{sec:cstruct}} 
\subsubsection{Proof of Proposition~\ref{p2}}\label{pp2}
\begin{proof}
	In order to prove the conservation equation we explicitly calculate the expression $\dot{J}=[H,J]$.
	Hence, we have the matrix equality  	\begin{align*}  
		\begin{pmatrix} 
			-	\dot{Z} &  \dot{Y} \\ \dot{D}   &   	\dot{C}
		\end{pmatrix}=	\begin{pmatrix} 
			D+Yw^2& C+Z+Yf \\w^2Z-fD+Cw^2  & -w^2Y+D+[C,f]
		\end{pmatrix}, 
	\end{align*} 
	where $C=Y^{-1}ZY$ and $D=-Y^{-1}(1+Z^2)$, where $Y:=BN^{-1}\sqrt{h}$ and $Z=-A$.
	The first two matrix entries are the first two conservation equations. To prove the other two we write out the other two entries and use the conservation equations 
	\begin{align*}
		\dot{D} &=   Y^{-1}\dot{Y}Y^{-1}(1+Z^2)-Y^{-1}(Z\dot{Z}+\dot{Z}Z)	\\&=
		w^2Z+fY^{-1}(1+Z^2)+Y^{-1}ZYw^2 .
	\end{align*}
	By inserting the time derivatives we have, 
	\begin{align*}
		& Y^{-1}(Y^{-1} Z Y+ Yf+ Z)Y^{-1}(1+Z^2)-Y^{-1}(Z (- Y w^2 +Y^{-1}(1+Z^2))+( - Y w^2 +Y^{-1}(1+Z^2))Z)
		\\&=
		w^2Z+fY^{-1}(1+Z^2)+Y^{-1}ZYw^2
	\end{align*}
	and \begin{align*}
		\dot{C}&=-Y^{-1}\dot{Y}Y^{-1}ZY+Y^{-1} \dot{Z} Y+Y^{-1}Z\dot{Y}
		\\&=-w^2Y-Y^{-1}(1+Z^2)+[Y^{-1}ZY,f] .
	\end{align*}
	Inserting the time derivatives we have, 
	\begin{align*}
		&-Y^{-1}(Y^{-1} Z Y+ Yf+ Z)Y^{-1}ZY+Y^{-1} (- Y w^2 +Y^{-1}(1+Z^2)) Y+Y^{-1}Z(Y^{-1} Z Y+ Yf+ Z)
		\\&=-w^2Y-Y^{-1}(1+Z^2)+[Y^{-1}ZY,f] .
	\end{align*}
	Both equations are trivially satisfied and supply no further information.
\end{proof} 

\subsection{Proofs of Section~\ref{sec:domains}}
		\stepcounter{subsubsection}\stepcounter{subsubsection}\stepcounter{subsubsection} 
		\subsubsection{Proof of Proposition~\ref{cssa}}\label{p41}
	
		\begin{proof}
			First we calculate the adjoint of the complex structure $J$   by using Equation~(\ref{adcs}) (which holds for the changed measure as well), i.e.
			\begin{align*}
			J_{Y}^{*}&= \varepsilon \,J_{Y}^{\dagger}\varepsilon^{T}\\
			&=\begin{pmatrix} 
			0&  1 \\ -1  & 0
			\end{pmatrix}\begin{pmatrix} 
			A^*& -(Y^{-1} (1+A^2))^*  \\  Y^* & -(Y^{-1}AY)^*
			\end{pmatrix}\begin{pmatrix} 
			0& - 1 \\ 1  & 0
			\end{pmatrix}  \\
			&=\begin{pmatrix} 
			0&  1 \\ -1  & 0
			\end{pmatrix}\begin{pmatrix} 
			-(Y^{-1} (1+A^2))^* & -A^*\\  -(Y^{-1}AY)^*  &- Y^* 
			\end{pmatrix}  \\
			&= \begin{pmatrix}  -(Y^{-1}AY)^* 
			&- Y^*  \\  (Y^{-1} (1+A^2))^* &A^*
			\end{pmatrix} \\
			&= \begin{pmatrix}  -(Y^{-1}AY)^* 
			&- Y\\  (Y^{-1} (1+A^2))^* &Y^{-1}AY
			\end{pmatrix} ,
			\end{align*}
			where in the last lines we used Assumption~\ref{ass1}. Next, we take a closer look at the remaining adjoint operators, 
			\begin{align*}
			(Y^{-1}AY)^* &=Y^*A^*(Y^{-1})^*\\&=
			Y(Y^{-1}AY)(Y^{*})^{-1}\\&=
			Y(Y^{-1}AY)(Y)^{-1}\\&=A,
			\end{align*}
			where we used the fact that $Y$ is invertible and self-adjoint (i.e.\ $Y=Y^{*}=Y^{**}$) and hence we can interchange the conjugation and the inverse operation since any power of a self-adjoint operator is by the spectral theorem self-adjoint (i.e.\ $(Y^{-1})^{*}=Y^{-1}=(Y^{*})^{-1}$). 
			The next adjoint term is, 
			\begin{align*}
			(Y^{-1} (1+A^2))^*&=(1+A^2)^*(Y^{-1})^*
			\\&=(1+(A^*)^2)(Y^*)^{-1}\\&=(1+(Y^{-1}AY)^2)(Y)^{-1}
			\\&=(1+Y^{-1}(A^2)Y)Y^{-1}\\&= Y^{-1} (1+A^2), 
			\end{align*}
			and therefore the equality $J_{ {Y}}^*=-J_{ {Y}}$ is proven. 
		\end{proof}
		
		\stepcounter{subsubsection}\stepcounter{subsubsection}
		\subsubsection{Proof of Lemma~\ref{noreq}}\label{l42}

	\begin{proof}
		To prove the equivalence of the norms we have to prove 
		\begin{align*}
		c_1\Vert\Psi\Vert_{W_V^2(\Sigma, {\mu})}^2\leq \Vert \Psi\Vert_{H_V^2}^2\leq 	 c_2\Vert\Psi\Vert_{W_V^2(\Sigma, {\mu})}^2,
		\end{align*}
		for positive real constants $c_1,\,c_2$ and vectors $\Psi\in W_V^2$. The first inequality is easily proven by setting $c_1=1$. The second holds as well since the operator $H_V$ is strictly positive i.e.\ we have
		\begin{align}\label{ineq1}
		\Vert H_V\Psi\Vert^2_{L^{2}(\Sigma,\,\mu)}\geq \epsilon  \Vert  \Psi\Vert^2_{L^{2}(\Sigma,\,\mu)}.
		\end{align}
		Hence, by setting the constant $c_2$ equal to $1+1/\epsilon$ we conclude the proof.
                
 		The   inclusion and density of $C_0^{\infty}(\Sigma)$ in $W_V^2(\Sigma, \mu)$ holds since in \cite{S01} it was proven that the domain $W_V^2$ is equivalent to the closure of the operator $H_V$ in $L^2$ from the initial domain $C_0^{\infty}$. The next inclusion follows from the definition  of the norm (see Equation~(\ref{norm})) of $W_V^2$, i.e.\ we have, 
		$$\Vert \Psi\Vert_{L^{2}(\Sigma,\,\mu)}\leq\Vert \Psi\Vert_{W_V^2(\Sigma, {\mu})}.$$
	\end{proof}
	
		\subsubsection{Proof of Lemma~\ref{l43a}}\label{pt43a}
	\begin{proof}
		Since the domain $W_V^{2s}$ is equivalent to the closure of the operator $H_V$ in $L^2$ from the initial domain $C_0^{\infty}$, density follows and therefore the first inclusion. For a different proof see also \cite[Proof of Theorem~4.2]{stl}, which holds in our case since the proofs in \cite{stl} have been kept general enough to hold for a positive operator.  In order to prove the next inclusion we first prove that  the norm $\Vert \cdot\Vert_{H_V^{s}(\Sigma, {\mu})}$ is equivalent to the following norms 
		\begin{equation}\label{norm3}
		\Vert \Psi\Vert_{W_V^{2s}(\Sigma, {\mu})}=\Vert H_V^{s}\Psi\Vert_{L^{2}(\Sigma,\,\mu)},\qquad 
		\qquad 	\Vert \Psi\Vert_{\mathcal{W}_V^{2s}(\Sigma, {\mu})}=\Vert (H_V+\rho)^{s}\Psi\Vert_{L^{2}(\Sigma,\,\mu)},
		\end{equation}
		where $\rho\geq1$.  
			The first equivalence is proven along the same lines as in the proof of Lemma~\ref{noreq}, i.e. let us first prove 
		\begin{align*}
		c_1\Vert\Psi\Vert_{W_V^{2s}(\Sigma, {\mu})}^2\leq \Vert \Psi\Vert_{H_V^{2s}(\Sigma, {\mu})}^2\leq 	 c_2\Vert\Psi\Vert_{W_V^{2s}(\Sigma, {\mu})}^2,
		\end{align*}
		for positive real constants $c_1,\,c_2$ and vectors $\Psi\in W_V^{2s}$. The first inequality holds by simply setting $c_1=1$ and hence we are left with proving, 
		\begin{align*}
		\Vert H_V^{s}\Psi\Vert^2_{L^{2}(\Sigma,\,\mu)}+\Vert  \Psi\Vert^2_{L^{2}(\Sigma,\,\mu)}\leq 	 c_2\Vert H_V^{s}\Psi\Vert^2_{L^{2}(\Sigma,\,\mu)}.
		\end{align*}
		Since the operator $H_V^{s}$ is strictly positive we have 
		\begin{align*}
		\Vert H_V^{s}\Psi\Vert^2_{L^{2}(\Sigma,\,\mu)}\geq\epsilon_s \Vert \Psi\Vert^2 _{L^{2}(\Sigma,\,\mu)},
		\end{align*}	
		for some finite positive constant $\epsilon_s$. Hence, by setting $c_2=1+1/\epsilon_s$ the inequality is satisfied. Next, we prove the equivalence of the norms  $\Vert \cdot\Vert_{W_V^{2s}(\Sigma, {\mu})}$ and $\Vert \cdot\Vert_{\mathcal{W}_V^{2s}(\Sigma, {\mu})}$ from which by transitivity the equivalence of  norms 
		$\Vert \cdot\Vert_{H_V^{s}(\Sigma, {\mu})}$ and $\Vert \cdot\Vert_{\mathcal{W}_V^{2s}(\Sigma, {\mu})}$ follows. The proof, as before, is done by means of the following inequalities, 
		\begin{align*}
		c_1\Vert\Psi\Vert_{W_V^{2s}(\Sigma, {\mu})}\leq \Vert \Psi\Vert_{\mathcal{W}_V^{2s}(\Sigma, {\mu})}\leq 	 c_2\Vert\Psi\Vert_{W_V^{2s}(\Sigma, {\mu})}.
		\end{align*}
		The first inequality is proven to hold by substitution, i.e.\ 
		\begin{align*}
		c_1\Vert H_V^{2s}\Psi\Vert_{L^{2}(\Sigma,\,\mu)} \leq \Vert  (H_V+\rho)^{2s}\Psi\Vert _{L^{2}(\Sigma,\,\mu)},
		\end{align*}
		we substitute $\Phi= (H_V+\rho)^{2s}\Psi$ which in turn gives us,
		\begin{align*}
		c_1\Vert H_V^{2s}(H_V+\rho)^{-2s}\Phi\Vert_{L^{2}(\Sigma,\,\mu)} \leq \Vert   \Phi\Vert _{L^{2}(\Sigma,\,\mu)}.
		\end{align*}
		By using the fact (\cite[Proof of Theorem X.12]{RS2} and \cite[Section 4]{stl}) that $\Vert H_V^{2s}(H_V+\rho)^{-2s}\Phi\Vert_{L^{2}(\Sigma,\,\mu)}\leq\Vert\Phi\Vert_{L^{2}(\Sigma,\,\mu)}$, for all $s>0$, and by setting $c_1=1$ the inequality is satisfied. The next inequality is proven in an analogous manner, i.e.\
		\begin{align*}
		\Vert \Phi\Vert_{L^{2}(\Sigma,\,\mu)} \leq c_2 \Vert  H_V^{2s}(H_V+\rho)^{-2s} \Phi\Vert_{L^{2}(\Sigma,\,\mu)} ,
		\end{align*}
		which reads 
		\begin{align*}\Vert  H \Phi\Vert_{L^{2}(\Sigma,\,\mu)} \geq 1/c_2
		\Vert \Phi\Vert_{L^{2}(\Sigma,\,\mu)}   ,
		\end{align*}
		where the operator $H=(H_V\,(H_V+\rho)^{-1})^{2s}$ is a positive operator and hence by definition   
		\begin{align*}\Vert  H \Phi\Vert_{L^{2}(\Sigma,\,\mu)} \geq \epsilon_c
		\Vert \Phi\Vert_{L^{2}(\Sigma,\,\mu)},
		\end{align*}
		for some $\epsilon_c>0$ and hence by choosing $c_2\geq 1/\epsilon_c+1$ the norm equivalence is concluded.
Equipped with the equivalence of the norms we next prove the remaining inclusions. For the second inclusion\footnote{It can also be proven by the use of the spectral theorem. We are indebted to Elmar Wagner for this remark.} to hold we have to prove the following inequality,
		\begin{align*}
		\Vert \Psi\Vert_{W_V^{2k}(\Sigma, {\mu})}^2\leq  \Vert \Psi\Vert_{W_V^{2s}(\Sigma, {\mu})}^2,
		\end{align*}
		since the norms $\Vert \cdot\Vert_{W_V^{2k}(\Sigma, {\mu})}$ and  $\Vert \cdot\Vert_{\mathcal{W}_V^{2k}(\Sigma, {\mu})}$ are equivalent we can prove the easier version 
		\begin{align*}
		\Vert \Psi\Vert_{\mathcal{W}_V^{2k}(\Sigma, {\mu})}^2\leq  \Vert \Psi\Vert_{\mathcal{W}_V^{2s}(\Sigma, {\mu})}^2.
		\end{align*}
		First note that since the operator $H_V$ is   strictly positive  and self-adjoint  we have for all positive $s$
		\begin{align}\label{res}
		\Vert (H_V+\rho)^{-s} \Vert\leq 1.
		\end{align}
		Next, we write $s=k+l$ to obtain 
		\begin{align*}
		\Vert \Psi\Vert_{\mathcal{W}_V^{2k}(\Sigma, {\mu})}^2&= \Vert (H_V+\rho)^{k}\Psi\Vert^2_{L^{2}(\Sigma,\,\mu)}
		\\&= \Vert (H_V+\rho)^{k}(H_V+\rho)^{-s}(H_V+\rho)^{s}\Psi\Vert^2 _{L^{2}(\Sigma,\,\mu)}
		\\& = \Vert (H_V+\rho)^{-s+k }(H_V+\rho)^{s }\Psi\Vert^2  \\& \leq \Vert  (H_V+\rho)^{s}\Psi\Vert^2 _{L^{2}(\Sigma,\,\mu)} =\Vert\Psi\Vert_{\mathcal{W}_V^{2s}(\Sigma, {\mu})}^2,
		\end{align*} 
		where in the last lines we used properties of a semigroup that is generated by
		$(H_V+\rho)^{-s}$. The semigroup behavior comes from the positivity of the operator $H_V$ and from the boundedness of its respective resolvent, see \cite[Section~4]{stl} and \cite[Proof of Theorem~X.31]{RS2}. Hence, in the first line we used the existence of the identity   and in the third line we used the  semigroup property $(H_V+\rho)^{k}(H_V+\rho)^{-s}=(H_V+\rho)^{k-s}$, if $s\geq k$ for positive $s$ and $k$.
		
		The second inclusion, $W_V^{2k}(\Sigma, \mu)\subset L^2(\Sigma, \mu)$, is easily proven and follows from the norm $\Vert \Psi\Vert_{H_V^{s}(\Sigma, {\mu})}$, i.e.
		\begin{align*}
		\Vert \Psi\Vert_{H_V^{s}(\Sigma, {\mu})}^2\geq\Vert \Psi\Vert^2_{L^{2}(\Sigma,\,\mu)}\Rightarrow \Vert \Psi\Vert_{H_V^{s}(\Sigma, {\mu})} \geq\Vert \Psi\Vert_{L^{2}(\Sigma,\,\mu)}.
		\end{align*}
	\end{proof}

        \stepcounter{subsubsection}\stepcounter{subsubsection}

\subsubsection{Proof of  Lemma~\ref{lscs2}}\label{l43}

\begin{proof}
	Anti self-adjointness of $J$ means that we have to prove the relation $J^*=-J$ on the Hilbert space $\mathscr{H}_{\overline{J}}$. Due to  Equation~(\ref{adcs}) we  write the adjoint as, 
	\begin{align*}
	J^*&=\epsilon J^{\dagger}\epsilon^T
	\\&=\epsilon (X^{-1}\overline{J}X)^{\dagger}\epsilon^T
	\\&=\epsilon X^{\dagger} \overline{J}^{\dagger}(X^{-1})^{\dagger}\epsilon^T
	\\&=\epsilon X^{\dagger} \epsilon^T \epsilon \overline{J}^{\dagger}\epsilon^T \epsilon (X^{-1})^{\dagger}\epsilon^T
	\\&=\epsilon X^{\dagger} \epsilon^T  \overline{J}^{*}  \epsilon (X^{-1})^{\dagger}\epsilon^T
	\\&=X^{-1} \overline{J}^{*}  X 
	\\&=-X^{-1} \overline{J}   X =-J,
	\end{align*}
	where in the last lines we used the fact that   $\epsilon \overline{J}^{\dagger}\epsilon^T=\overline{J}^{*}=-\overline{J}$ and that $\epsilon \,X^{\dagger}\,\epsilon^T=  \pm   X^{-1}$ implies $\epsilon (X^{-1})^{\dagger}\epsilon^T=  \pm  X$.
\end{proof}

\subsection{Proofs of Section~\ref{sec:natcomplex}}

\stepcounter{subsubsection}\stepcounter{subsubsection}\stepcounter{subsubsection}\stepcounter{subsubsection}

\subsubsection{Proof of Theorem~\ref{t1}}\label{pt1}
 		\begin{proof}
 		The function $f=-N^{-1}\partial_{t}N$ (see Equation~(\ref{eq:kgt})) reduces Conservation equation~(\ref{me3}) to
 		\begin{align*} 
 		\partial_t^2 Y -(\partial_tY)f-Y(\partial_tf) =
 		- 2Y w^2 +2Y^{-1}(1+\frac{1}{4}(\partial_tY -Yf)^2),
 		\end{align*}
 		which reads 
 		\begin{align*} 
 		\partial_t^2 Y -(\partial_tY)f-Y(\partial_tf) =
 		- 2Y w^2 +\frac{1}{2}Y^{-1}(4+(\partial_tY)^2 -2(\partial_tY)\,Yf+Y^2f^2),
 		\end{align*}
 		which reduces to 
 		\begin{align*}  0=&\,
 		Y\,\partial_t^2 Y -Y^2(\partial_tf) 
 		+ 2Y^2 w^2 -\frac{1}{2} (4+(\partial_tY)^2 +Y^2f^2)\\=& 
 		Y\,\partial_t^2 Y
 		-\frac{1}{2} (\partial_tY)^2
 		-Y^2(\partial_tf) 
 		+ 2Y^2 w^2 -\frac{1}{2} (4+ Y^2f^2)\\=& 
 		Y\,\partial_t^2 Y
 		-\frac{1}{2} (\partial_tY)^2
 		-Y^2\left(
 		\partial_tf +\frac{1}{2}f^2 -
 		2  w^2
 		\right)-2 .
 		\end{align*}  
 	  Next, we perform the variable substitution $Y=B_N N^{-1}$ and the previous non-linear and non-homogeneous differential equation reads  
 	 \begin{align*}
 	 \left((N^{-1})\partial_t^2 (N^{-1}) -\frac{1}{2}(\partial_t(N^{-1}) )^2-(N^{-1})^2\left(
 	 \partial_tf +\frac{1}{2}f^2 -
 	 2  w^2
 	 \right)\right)	B_{N}^2-2=2b(\partial_t 	B_{N} ,\partial_t^2 	B_{N} ),
 	 \end{align*}
 which after inserting the explicit expression of the term $\partial_tf +\frac{1}{2}f^2$ 
 	  \begin{align*}
 	  \partial_tf +\frac{1}{2}f^2&=  N^{-2} ( \partial_t N )^2-N^{-1}\partial_t^2 N
 	  +\frac{1}{2}\,
 	  N^{-2} ( \partial_t N )^2 \\&= \frac{3}{2} N^{-2} ( \partial_t N )^2 -N^{-1}\partial_t^2 N
 	  \end{align*}  	reads
 		\begin{align*}
 		b(\partial_tB_{N},\,\partial_t^2B_{N} )&=-2B_{N}^2\,N^{-2} \left( 
 		w^2
 		\right)+2\\&=-2B_{N}^2\,  \left(-
 		\frac{1}{\sqrt{|	\underline{h}|}}\partial_i(
 		\sqrt{|	\underline{h}|}   	\underline{h}^{ij}\partial_j
 		)+\xi\,{R}+m^2
 		\right)+2.
 		\end{align*}
 	Since the right-hand side, after the cancellations,  does not depend on the time we obtain the following solution
 		\begin{align*}
 		B_{N}=\left(-\Delta_{\underline{h}}+\xi\,{R}+m^2
 		\right)^{-1/2}.
 		\end{align*}The proof of  anti-self-adjointness and the isomorphism of solution spaces of the complex structure is equivalent to the proof in Theorem~\ref{t0}. 
 	\end{proof}
 	
 	\subsubsection{Proof of Lemma~\ref{tcts}}\label{ptcts}
 	
 	\begin{proof}	To see that the conformal transformation leaves the symplectic structure invariant we use  Transformation~(\ref{ct}) and  Transformation~(\ref{ctf}) and we denote by ${h}$ the determinant of the spatial part of the metric $\mathbf{g}$ and by $\underline{{h}}$ the determinant  of the spatial part of the metric  $\overline{\mathbf{g}}$. 
 		\begin{align*}\Omega  ( \overline{\Phi}_1, \overline{\Phi}_2) &=
 		\int_{\Sigma}\left(
 		\overline{\pi}_1 \overline{\varphi}_2 - \overline{\pi}_2\overline{\varphi}_1 
 		\right) d^3x \\&=
 		\int_{\Sigma}\left(
 		\pi_1 \varphi_2 
 		+N^{-1}a^2\partial_{t}a\sqrt{| \underline{h}|}\varphi_1\varphi_2
 		- \pi_2\varphi_1 -N^{-1}a^2\partial_{t}a\sqrt{| \underline{h}|}\varphi_2\varphi_1
 		\right) d^3x\\&= 
 		\Omega ( \Phi_1, \Phi_2),
 		\end{align*} 	
 		where in the last lines we used the explicit expressions of the Cauchy data and took into account their conformal transformation, i.e. 
 		\begin{align*}\overline{\pi}
 		&=\overline{N}^{-1}  \sqrt{	\underline{h}}\,\partial_{t}\overline{\phi}\\&=\overline{N}^{-1}  \sqrt{	\underline{h}}\, \partial_{t}(a \,\phi)\\&=
 		\overline{N}^{-1} \,\sqrt{	\underline{h}}\,  (\partial_{t}a)\varphi+\overline{N}^{-1} \,\sqrt{	\underline{h}} a \,\partial_{t}\phi\\&=  
 		\overline{N}^{-1} \,\sqrt{	\underline{h}}\,  (\partial_{t}a)\varphi+  a^{-1}\pi\\&=  
 		\overline{N}^{-1} \,\sqrt{	\underline{h}}\,  (\partial_{t}a)\varphi+  a^{-1}\pi.
 		\end{align*} 
 		By using this equality, transformations of vectors of the canonical phase space $\Gamma_t$ are straightforwardly obtained,
 		\begin{align*} 
 		\begin{pmatrix}\overline{\varphi}\\\overline{\pi}\end{pmatrix}=
 		\begin{pmatrix}a\,{\varphi}\\  
 		\overline{N}^{-1} \sqrt{	\underline{h}}\,  (\partial_{t}a)\varphi+  a^{-1}\pi\end{pmatrix}=\begin{pmatrix} 
 		a&  0 \\ 	\overline{N}^{-1} \,\sqrt{	\underline{h}}\,  (\partial_{t}a)    &  a^{-1}
 		\end{pmatrix} 
 		\begin{pmatrix}{\varphi}\\{\pi}\end{pmatrix}.
 		\end{align*}
	\end{proof}
 
\subsubsection{Proof of Lemma~\ref{tctsb}}\label{ptctsb}

\begin{proof}
  The similarity of the complex structures is proven by using the symplectic structures and in particular the auxiliary scalar product defined in Equation~(\ref{ssm}),
 		\begin{align*}
 		\Omega  ( \overline{\Phi}_1,\overline{J}\, \overline{\Phi}_2) &=
 		\langle\overline{\Phi}_1,\varepsilon \overline{J}\,\overline{\Phi}_2\rangle\\& =  
 		\langle X{\Phi}_1,\varepsilon \overline{J}\,X{\Phi}_2\rangle\\& =  
 		\langle {\Phi}_1,  X^{T}\varepsilon\overline{J}\,X{\Phi}_2\rangle\\& = 
 		\langle {\Phi}_1, \varepsilon\,\varepsilon^T X^{T}\varepsilon\overline{J}\,X{\Phi}_2\rangle\\& = 	\Omega ( \Phi_1, \varepsilon^T X^{T}\varepsilon\overline{J}\,X\Phi_2)\\& = 	\Omega ( \Phi_1, X^{-1}\overline{J}\,X\Phi_2)\\& = 	\Omega ( \Phi_1, J\Phi_2),
 		\end{align*}
 		where we deduced from a simple calculation that $\varepsilon^T X^{T}\varepsilon=X^{-1}$. Next, we prove the similarity of the complex structures w.r.t.\ the measure $\mu$, 
 		\begin{align*}
 		J_{Y}&=TJT_{Y}^{-1}\\&=
 		TX\overline{J}X^{-1}T^{-1}
 		\\&=TXT^{-1}T \overline{J}T^{-1}T X^{-1}T^{-1}\\&=
 		TXT^{-1}\overline{J}_{Y}T X^{-1}T^{-1},
 		\end{align*}
 		where in the last lines we used Lemma~\ref{lcsme} and Equation~(\ref{tme}). 
 	\end{proof}
 	 
\subsubsection{Proof  of Theorem~\ref{t4}}\label{t54}
\begin{proof}
The function $f$, given by $f=-N^{-1} ( \partial_t N )+3 a  ^{-1}  \,(\partial_ta)$  reduces Conservation equation~(\ref{me3}) to
		\begin{align*} 
		\partial_t^2 Y -(\partial_tY)f-Y(\partial_tf) =
		- 2Y w^2 +2Y^{-1}(1+\frac{1}{4}(\partial_tY -Yf)^2),
		\end{align*}
		which in turn reduces, as before (see Proof~\ref{pt1}) to 
		\begin{align} \label{eq:diffeq}
		Y\,\partial_t^2 Y
		-\frac{1}{2} (\partial_tY)^2
		-Y^2\left(
		\partial_tf +\frac{1}{2}f^2 -
		2  w^2
		\right)-2   =0 .
		\end{align} 
		Next, we perform the substitution $Y=B_{N}\,N^{-1}\,a$
			\begin{align*}
		\left((N^{-1}a)\partial_t^2 (N^{-1}a) -\frac{1}{2}(\partial_t(N^{-1}a) )^2-(N^{-1}a)^2\left(
		\partial_tf +\frac{1}{2}f^2 -
		2  w^2
		\right)\right)	B_{N}^2-2=2b(\partial_t 	B_{N} ,\partial_t^2 	B_{N} ),
		\end{align*}	
	and obtain the differential equation
		\begin{align*}
		b(\partial_tB_{N},\,\partial_t^2B_{N} )&=2B_{N}^2\,N^{-2}a^2\left( \underbrace{a^{-2}(\partial_t a )^2-(N^{-1}\partial_tN )\,( a^{-1}\,\partial_ta)+a^{-1}\partial_t^2a}_{\frac{1}{6}N^2\,R-\frac{1}{6}N^2\,a^{-2}\underline{R}}
		-w^2
		\right)+2\\&=2B_{N}^2\,  \left(
		\frac{1}{\sqrt{|	\underline{h}|}}\partial_i(
		\sqrt{|	\underline{h}|}   	\underline{h}^{ij}\partial_j
		)-\frac{1}{6}\,\underline{R}
		\right)\\&=2B_{N}^2\,  \left(
		\Delta_{\underline{h}}-\frac{1}{6}\,	\underline{h}^{ij}\underline{R}_{ij}
		\right)+2,
		\end{align*}
		where the curvature scalar of the metric under consideration is given by (see Remark~\ref{remcurv}),
		\begin{align*}
		\frac{1}{6} N^2\, R&=\frac{1}{6} N^2\,R_{00}g^{00}+\frac{1}{6} N^2\,R_{ij}g^{ij}\\&= a^{-2}	(\partial_t  a)^2-a^{-1}	N^{-1}(\partial_t  a)		(\partial_t  N)+ a^{-1}	\, \partial_t^2  a +\frac{1}{6} N^2\,\underline{R}.
		\end{align*}  Hence,  the remaining curvature term does not depend on the time  and the solution for $B_{N}$ is given as,
		\begin{equation*}
		B_{N}= \left(-
		\Delta_{\underline{h}}+\frac{1}{6}\,	 \underline{R} 
		\right)^{-1/2},
		\end{equation*}
		and therefore $Y=\left(-
		\Delta_{\underline{h}}+\frac{1}{6}\,\underline{R} 
		\right)^{-1/2}\,N^{-1}\,a.$  
		
		Next, we prove that the complex structure is anti-self-adjoint and an automorphism. Due to the similarity of the complex structure $J_{\mu}$ to the one evaluated in Theorem~\ref{t1} we only need to prove that 
		\begin{align*}
		\epsilon \,X_{T}^{\dagger}\,\epsilon^T=     X^{-1}_{T} ,
		\end{align*}
		where $X_{T}:=T\,X \,T^{-1}$. If this holds, the proof of anti self-adjointness follows by Lemma~\ref{lscs2}	and Lemma~\ref{tcts} and Lemma~\ref{tctsb}. A simple matrix calculation proves that this  equality holds. We  turn next  to the automorphism property of the constructed complex structure, i.e. we have to prove that  
		\begin{align*}
		\Vert J_{Y}\Psi\Vert_{W_V^{2s+1}(\Sigma, {\mu})\times W_V^{2s}(\Sigma, {\mu})}^2\leq
		c\Vert  \Psi\Vert_{W_V^{2s+1}(\Sigma, {\mu})\times W_V^{2s}(\Sigma, {\mu})}^2.
		\end{align*} 
		By using the similarity property we have, 	 
		\begin{align*}
		\Vert J_{Y}\Psi\Vert_{W_V^{2s+1}(\Sigma, {\mu})\times W_V^{2s}(\Sigma, {\mu})}^2 & =   
		\Vert X_{T}^{-1}\overline{J}_{Y}X_{T} \Psi\Vert_{W_V^{2s+1}(\Sigma, {\mu})\times W_V^{2s}(\Sigma, {\mu}) }^2  \\& =  
		\Vert {T}\,	X^{-1}\,{T}^{-1}\,   \overline{J}_{Y}\, {T}\,X \,{T}^{-1}\Psi\Vert_{W_V^{2s+1}(\Sigma, {\mu})\times W_V^{2s}(\Sigma, {\mu})}^2  ,
		\end{align*} 
		where the explicit form of the matrix ${T}$ is given in Equation~(\ref{tmat}).	The  transformed similarity matrices are given by, 
		\begin{align*}
		{T}\,	X^{-1}\,{T}^{-1}=\begin{pmatrix} 
		a^{-1} &  0 \\ - (\partial_{t}a)    &    a
		\end{pmatrix},\qquad\qquad {T}\,	 X\,{T}^{-1}=\begin{pmatrix} 
		a &  0 \\   (\partial_{t}a)    &    a^{-1}
		\end{pmatrix},
		\end{align*} 		
		and hence they only depend on time, which gives us
		\begin{align*} 
		\Vert J_{Y}\Psi\Vert_{W_V^{2s+1}(\Sigma, {\mu})\times W_V^{2s}(\Sigma, {\mu})}^2 &=	\Vert {T}\,	X^{-1}\,{T}^{-1}\,  \overline{J}_{Y}\, {T}\,X \,{T}^{-1}\Psi\Vert_{W_V^{2s+1}(\Sigma, {\mu})\times W_V^{2s}(\Sigma, {\mu})}^2  \\&=	\Vert    \overline{J}_{Y} \Psi\Vert_{W_V^{2s+1}(\Sigma, {\mu})\times W_V^{2s}(\Sigma, {\mu})}^2 \\& =
		\Vert H_V^{-1/2} \overline{\mu}^{-1}{\mu}\Psi_2\Vert_{W_V^{2s+1}(\Sigma,  {\mu})}^2 + \Vert{\mu}^{-1} \overline{\mu}  H_V^{1/2}\Psi_1\Vert_{W_V^{2s}(\Sigma,  {\mu})}^2 \\&   =
		\Vert H_V^{-1/2} a^{2}\Psi_2\Vert_{W_V^{2s+1}(\Sigma,  {\mu})}^2 + \Vert a^{-2}  H_V^{1/2}\Psi_1\Vert_{W_V^{2s}(\Sigma,   {\mu})}^2 \\&\nonumber\leq
		c_4 \Vert \Psi_2\Vert_{W_V^{2s}(\Sigma, {\mu})}^2 + c_5\Vert \Psi_1\Vert_{W_V^{2s+1}(\Sigma, {\mu})}^2\\&
		\leq
		c\Vert  \Psi\Vert_{W_V^{2s+1}(\Sigma, {\mu})\times W_V^{2s}(\Sigma, {\mu})}^2,
		\end{align*} 		
		where we used the fact that   the transformed complex structure $\overline{J}_{Y}$ and the matrix 
		${T}\,X \,{T}^{-1}$  commute. Moreover, the scalar product was taken w.r.t.\ the measure $\mu$ which is related to the measure $\overline{\mu}$, i.e.\ the measure w.r.t.\ the metric $\overline{g}$, by a factor $\mu=a^2\overline{\mu}$. Since $a$ is purely time dependent the measure changes merely correspond to constant factors. 
		
	\end{proof}

	\subsubsection {Proof of Theorem~\ref{t6}}\label{t55}
	\begin{proof}
	We have by using the transformation $Y=B_{N}\,N^{-1}\,a$ the following differential equation that follows from Conservation equation~(\ref{me3}),
		\begin{align*}
		b(\partial_tB_{N},\,\partial_t^2B_{N} )&=2B_{N}^2\,N^{-2}a^2\left( \underbrace{a^{-2}(\partial_t a )^2-(N^{-1}\partial_tN )\,( a^{-1}\,\partial_ta)+a^{-1}\partial_t^2a}_{\frac{1}{6}N^2\,R-\frac{1}{6}a^{-2}N^2\,\underline{R}}
		-w^2
		\right)+2\\&=2B_{N}^2\, a^2\left(
		\frac{1}{\sqrt{|	\underline{h}|}}\partial_i(
		\sqrt{|	\underline{h}|}   	 {h}^{ij}\partial_j
		)-\frac{1}{6}\,a^{-2}\underline{R}-a^{-2}m^2
		\right)\\&=2B_{N}^2\,  \left(
		\Delta_{\underline{h}}-\frac{1}{6}\,	\underline{h}^{ij}\underline{R}_{ij}- m^2
		\right)+2,
		\end{align*}
		and the solution for $B_{N}$ is given as,
		\begin{equation*}
		B_{N}= \left(-
		\Delta_{\underline{h}}+\frac{1}{6}\,	\underline{h}^{ij}\underline{R}_{ij}+ m^2
		\right)^{-1/2},
		\end{equation*}
		and therefore $Y=\left(-
		\Delta_{\underline{h}}+\frac{1}{6}\,	 \underline{R} + m^2
		\right)^{-1/2}\,N^{-1}\,a$.		The proof of the anti-self-adjointness and the isomorphism of the complex structure is equivalent to the proof of Theorem~\ref{t4}.	\end{proof}

	\subsubsection {Proof of Theorem~\ref{t7}}\label{t66}
	\begin{proof} 
	As before, if $f$ is independent of the spatial directions the Conservation equation  (\ref{me3}) reduces to  a single  differential equation
	\begin{align}  
	Y\,\partial_t^2 Y
	-\frac{1}{2} (\partial_tY)^2
	-Y^2\left(
	\partial_tf +\frac{1}{2}f^2 -
	2  w^2
	\right)=2,
	\end{align}  
	where the explicit form of the term $\partial_tf +\frac{1}{2}f^2	-2	w^2$ is given by, 
		\begin{align*}
	\partial_tf +\frac{1}{2}f^2	-2	w^2&=\frac{3}{2} N^{-2} ( \partial_t N )^2-N^{-1}\partial_t^2 N
	+2 {N^2a^{-2}}\Delta_{\underline{h}}	-2N^2m^{2}-2 \xi a^{-2} N^2\,\underline{R}
	\\&+3( {1}/{2}-4\xi)a^{-2} ( \partial_t a )^2+3(1-4\xi)a^{-1}  \,(\partial^2_t a)\\&-3(1-4\xi)  (N^{-1} \partial_t N )  (  a^{-1}\partial_t a ).
	\end{align*} 	
	Next, we insert the explicit form of the former term and substitute	$Y=\Lambda 	B_{N}$ in order to obtain the equation
	\begin{align*}
	\left(\Lambda\partial_t^2 \Lambda -\frac{1}{2}(\partial_t\Lambda )^2-\Lambda^2\left(
	\partial_tf +\frac{1}{2}f^2 -
	2  w^2
	\right)\right)	B_{N}^2-2=2b(\partial_t 	B_{N} ,\partial_t^2 	B_{N} ),
	\end{align*}	
  by setting $\Lambda=N^{-1}a$, we have for the left side
	\begin{align*}
	\left((1-6\xi)N^{-3}a\partial_tN (\partial_ta)-(1-6\xi) N^{-2}(\partial_ta)^{2}-(1-6\xi)N^{-2}a\,\partial_t^2 a +k^2+a^2\,M^2\right)	B_{N}^2-1,
	\end{align*}				
	where $k^2=-\Delta_{\underline{h}}+\xi\underline{R}$.	 
\end{proof}

\subsubsection {Proof of Example~\ref{lem:11b}}\label{t88b}
\begin{proof}
	The proof is done analogous to the one of Example~\ref{lem:10} (see proof below first). As before, we first calculate   explicitly the term $\partial_tf +\frac{1}{2}f^2	-2	w^2$ for the metric $g=-  dt^2+a^{2}(t)	\underline{h}_{ij}(\vec{x})\,dx^idx^j$, with $a=t^{1/2}$, 
	\begin{align*}
	\partial_tf +\frac{1}{2}f^2	-2	w^2&= 
	+2 { a^{-2}}\Delta_{\underline{h}}	-2 \xi a^{-2}  \,\underline{R}-2 m^{2}
	+3( {1}/{2}-4\xi)a^{-2} ( \partial_t a )^2+3(1-4\xi)a^{-1}  \,(\partial^2_t a) 
	\\& = -2\left(m^2+t^{-1} 
	H_V + \frac{3}{16}t^{-2}\right),
	\end{align*} 
	where we defined the operator $H_V :=-\Delta_{\underline{h}}	+\xi   \,\underline{R}$. As in the previous proofs, we solve Differential  Equation~(\ref{eq:ddiffeq}) by standard methods. 
\end{proof}

\subsubsection{Proof of Example~\ref{lem:10}}\label{t77}
	\begin{proof} 	
	First we insert  the metric $g=-  a^{2}(t) dt^2+a^{2}(t)	\underline{h}_{ij}(\vec{x})\,dx^idx^j$, with $a(t)=e^{Ht}$, to obtain the explicit term, 		\begin{align*}	
			\partial_tf +\frac{1}{2}f^2	-2	w^2&=\frac{3}{2} a^{-2} ( \partial_t a )^2-a^{-1}\partial_t^2a
		+2  \Delta_{\underline{h}}	-2a^2m^{2}-2 \xi  \,\underline{R}
		\\&+3( {1}/{2}-4\xi)a^{-2} ( \partial_t a )^2+3(1-4\xi)a^{-1}  \,(\partial^2_ta)\\&-3(1-4\xi)     (  a^{-2}( \partial_t a )^2
		 )	\\&=  
		 +2  \Delta_{\underline{h}}	-2 \xi  \,\underline{R}-2a^2m^{2}
		  +2(1-6\xi)a^{-1}  \,(\partial^2_t\,a) \\&=  
		  -2\left(H_V  +a^2\,m^{2}
		   \right)
		 	\end{align*}  
		 	with operator $H_V := (-\Delta_{\underline{h}}+\xi\underline{R}+ (6\xi-1)H^2)$.
		 By differentiating the Differential Equation~(\ref{eq:diffeq}), i.e.\ 
		 \begin{align}  \label{eq:ddiffeq}
		 \partial_t^3 Y 
		 -2  (\partial_t  Y) \left(
		 \partial_tf +\frac{1}{2}f^2 -
		 2  w^2
		 \right)-  Y  \left(
		 \partial_t^2f + f\partial_t f -
		 2  \partial_t(w^2)
		 \right)=0,
		 \end{align}
	and inserting the explicit term, calculated above, we obtain a general solution for $Y$   
		\begin{align*}	
		Y(t)& = c_1 \,_1 F_2\biggl(\frac{1}{2};1 - \frac{ iH_V^{1/2}}{ H}, \frac{ iH_V^{1/2}}{ H} + 1;-H^{-2  } m^{2} e^{ -2H t }\biggr)
		\\&
		+ c_2 \,  ( H^{-2  } m^{2} e^{ -2H t } )^{-iH_V^{1/2}/H}   \,_1 F_2\biggl(\frac{1}{2} - \frac{ iH_V^{1/2}}{H};1 - \frac{2 iH_V^{1/2}}{  H}, 1 - \frac{ iH_V^{1/2}}{H};-H^{-2  } m^{2} e^{ -2H t }\biggr)
		\\& + c_3\,   ( H^{-2  } m^{2} e^{  2H t } )^{ iH_V^{1/2}/H}   \,_1F_2\biggl(\frac{ iH_V^{1/2}}{H} + \frac{1}{2};\frac{ iH_V^{1/2}}{H} + 1, \frac{2 iH_V^{1/2}}{ H} + 1;-H^{-2  } m^{2} e^{ -2H t }  \biggr)		 	,
		\end{align*} 
	 where $\,_1 F_2$ is  a generalized hypergeometric function and $c_1,\,c_2,\,c_3$ are real valued constants. Next, we use the following special values of the hypergeometric functions  
		\begin{align*}	
		\,_1 F_2\biggl(\frac{1}{2};	\beta, 2-	\beta;z\biggr)&=
		\pi(1-	\beta)\,csc(\pi\beta)I_{1-	\beta}(\sqrt{z})\,I_{	\beta-1}(\sqrt{z})
		\\
		\,_1 F_2\biggl(	\beta;2	\beta, 	\beta+\frac{1}{2};z\biggr)&= 
		2^{2	\beta-1}\,z^{1/2-	\beta}\,\Gamma(	\beta+1/2)^2\,I_{	\beta-1/2}(\sqrt{z})^2,
		\end{align*} 
		where $I_{\nu}(\tau)$ denotes the modified Bessel functions of the first kind. 
		In order to ease the following calculations we define the variables $\alpha:=iH_V^{1/2}/H$ and $-z^2:=-{e^{ 2 H t}m^2}/{(   H^2) } $  which turns the solution into, 
		\begin{align*}	
		Y & 
		=c_1 \,\pi\,\alpha\, {csc(\pi(1-\alpha))}
		I_{\alpha}(i {z})\,I_{-\alpha}(i{z})\\& +	c_2\, 
		2^{-2\alpha}  \,\Gamma(1-\alpha)^2\,I_{-\alpha}(i{z})^2
		+c_3 \,2^{ 2\alpha}\,
		\Gamma(1+\alpha)^2\,I_{\alpha}(i{z})^2,		\end{align*} 
		by using Euler's reflection formula, i.e.\ 	 $\pi \,csc(\pi x)=\Gamma(x)\,\Gamma(1-x)$ (for $x\in\mathbb{C}-\mathbb{Z}$), we rewrite the former solution as,  
		\begin{align*}	
		Y & 
		=c_1 \, \alpha\, \Gamma(1-\alpha)\Gamma(\alpha)
		I_{\alpha}(i{z})\,I_{-\alpha}(i{z})\\& +	c_2\, 
		2^{-2\alpha}  \,\Gamma(1-\alpha)^2\,I_{-\alpha}(i{z})^2
		+c_3 \,2^{ 2\alpha}\,
		\Gamma(1+\alpha)^2\,I_{\alpha}(i{z})^2 .	\end{align*} 
		Since the variable $z$ is strictly negative we use the identity $I_{\alpha}(i\gamma)=i^{\alpha}J_{\alpha}( \gamma)$ and $\Gamma(\alpha+1)=\alpha\Gamma(\alpha)$ to obtain   
		\begin{align*}	
		Y & 
		=c_1 \,   \Gamma(1-\alpha)\Gamma(1+\alpha)
		J_{\alpha}( {z})\,J_{-\alpha}( {z})\\& +	c_2\, e^{-i\pi\alpha}
		4^{- \alpha}  \,\Gamma(1-\alpha)^2\,J_{-\alpha}( {z})^2
		+c_3 \,4^{  \alpha}\, e^{ i\pi\alpha}
		\Gamma(1+\alpha)^2\,J_{\alpha}( {z})^2 
		\end{align*} 
		 and due to the positivity condition on $Y$ (see Equations~(\ref{sop})) we redefine the constants $c_2$ and $c_3$ in order to eliminate the exponential factors.	\end{proof}
	  
	\subsubsection {Proof of Example~\ref{lem:11}}\label{t88}
\begin{proof}
	The proof is done analogous to the one of Example~\ref{lem:10}. As before, we first calculate   explicitly the term $\partial_tf +\frac{1}{2}f^2	-2	w^2$ for the metric $g=-  dt^2+a^{2}(t)	\underline{h}_{ij}(\vec{x})\,dx^idx^j$, with $a(t)=e^{Ht}$, 
 
		\begin{align*}
\partial_tf +\frac{1}{2}f^2	-2	w^2&= 
+2 { a^{-2}}\Delta_{\underline{h}}	-2 \xi a^{-2}  \,\underline{R}-2 m^{2}
\\&+3( {1}/{2}-4\xi)a^{-2} ( \partial_t a )^2+3(1-4\xi)a^{-1}  \,(\partial^2_t a) 
\\& = 
+2 { a^{-2}}\Delta_{\underline{h}}	-2 \xi a^{-2}  \,\underline{R}-2 m^{2}
 +3\, (  {3}/{2}-8\xi)H^2\\& = -2\left(a^{-2} 
H_V^2 + M^2\right) 
\end{align*} 
 where we defined the operator $H_V :=-\Delta_{\underline{h}}	+\xi   \,\underline{R}$ and the shifted mass variable $M^2:=m^{2} +2\,( 6\xi  -\frac{9}{8}   )H^2$. By solving Differential  Equation~(\ref{eq:ddiffeq}) and performing analogous steps as in Proof~\ref{t77}, we obtain the general solution 
 	\begin{align*}	
 Y & 
 =c_1 \,   \Gamma(1-\alpha)\Gamma(1+\alpha)
 J_{\alpha}( {z})\,J_{-\alpha}( {z})\\& +	c_2\,  
 4^{- \alpha}  \,\Gamma(1-\alpha)^2\,J_{-\alpha}( {z})^2
 +c_3 \,4^{  \alpha}\, 
 \Gamma(1+\alpha)^2\,J_{\alpha}( {z})^2 ,
 \end{align*} 
 with variables $\alpha:=-iM/H$ and $-z^2:=-{e^{- 2 H t}H_V }/{(H^2)}$.
\end{proof}

 	\subsubsection {Proof of Example~\ref{lem:12}}\label{t88c}
\begin{proof} 	
	First we insert  the metric $g=-  a^{2}(t) dt^2+a^{2}(t)	\underline{h}_{ij}(\vec{x})\,dx^idx^j$, with $a(t)=r/t$, to obtain the explicit term, 		\begin{align*}	
	\partial_tf +\frac{1}{2}f^2	-2	w^2&=  
	+2  \Delta_{\underline{h}}	-2 \xi  \,\underline{R}-2a^2m^{2}
	+2(1-6\xi)a^{-1}  \,(\partial^2_t\,a) \\&=  
	-2\left(H_V^2  +t^{-2}\,(r^2\,m^{2}-2(1-6\xi))
	\right)
	\end{align*}  
	with operator $H_V := (-\Delta_{\underline{h}}+\xi\underline{R})$.
Analogous to the previous proofs, we solve Differential  Equation~(\ref{eq:ddiffeq}) by standard methods. 	\end{proof}

\addtocontents{toc}{\endgroup}